\documentclass[electronics,article,submit,moreauthors,pdftex,10pt,a4paper]{Definitions/mdpi} 
\usepackage{graphicx} 
\newtheorem{lemma}[theorem]{Lemma}

\usepackage{algorithm,algorithmic,setspace}
\usepackage{amsmath}
\usepackage{amsthm,siunitx,amssymb,wasysym,array,url, color,subfigure}
\newcommand{\hll}{\color{red}}
\firstpage{1} 
\pubvolume{xx}
\issuenum{1}
\articlenumber{1}
\pubyear{2018}
\copyrightyear{2018}
\externaleditor{Academic Editor: name}
\history{Received: date; Accepted: date; Published: date}

\Title{SBL-Based Direction Finding Method With Imperfect Array}

\Author{Peng Chen  $^{1}$\orcidA{}, Zhimin Chen $^{2,}$*, Xuan Zhang $^{3}$ and Linxi Liu $^{3}$}

\AuthorNames{Peng Chen, Zhimin Chen, Xuan Zhang and Linxi Liu}

\address{%
$^{1}$ \quad State Key Laboratory of Millimeter Waves, Southeast Univerity, Nanjing 210096, China; chenpengseu@seu.edu.cn\\
$^{2}$ \quad School of Electronic and Information Engineering, Shanghai Dianji University, Shanghai 201306, China;\\
$^{3}$ \quad School of Information Science and Engineering, Southeast University, China; \{zhangxuan, liulinxi\}@seu.edu.cn}

\corres{Correspondence: chenzm@sdju.edu.cn; Tel.: +89-181-1631-4602}

\abstract{The imperfect array degrades the direction finding performance. In this paper, we investigate the direction finding problem in uniform linear array (ULA) system with unknown mutual coupling effect between antennas. By exploiting the target sparsity in the spatial domain, sparse Bayesian learning (SBL)-based model is proposed and converts the direction finding problem into a sparse reconstruction problem. In the sparse-based model, the \emph{off-grid} errors are introduced by discretizing the direction area into grids. Therefore, an off-grid SBL model with mutual coupling vector is proposed to overcome both the mutual coupling and the off-grid effect. With the distribution assumptions of unknown parameters including the noise variance, the off-grid vector, the received signals and the mutual coupling vector, a novel direction finding method based on SBL with unknown mutual coupling effect named DFSMC is proposed, where an expectation-maximum (EM)-based step is adopted by deriving the estimation expressions for all the unknown parameters theoretically. Simulation results show that the proposed DFSMC method can outperform state-of-the-art direction finding methods significantly in the array system with unknown mutual coupling effect.}


\begin{document}


\section{Introduction} \label{sec1}

{\hll In the direction finding problem, the traditional discrete Fourier transform (DFT)-based method can only find one signal in one beam-width,  so the resolution of such a method is too low to estimate multiple signals.} Therefore, the super-resolution methods have been proposed including multiple signal classification (MUSIC) method~\cite{ralph1986,schmidt1981}, Root-MUSIC method~\cite{Zoltowski1993}, and the estimating signal parameters via rotational invariance techniques (ESPRIT) method~\cite{roy1989}. {\hll Additionally, the subspace methods have also been improved to estimate the correlation signals,} such as the spatial smoothing MUSIC method~\cite{Pham2017}. However, the subspace methods only distinguish the noise and signal subspaces and have not exploited {\hll additional} characteristics of the received signals.

The compressed sensing (CS)-based methods {\hll have been} proposed to estimate the directions by exploiting the signal sparsity in the spatial domain~\cite{pengele,Matteo2013,7467561,yao2011,Matteo2016,Chen:2017ena,yang2017SP,Qing2016,Yang20166}. Notably, the sparse Bayesian learning (SBL) and the relevance vector machine (RVM) proposed in~\cite{Tipping2001} can achieve better estimation performance in the CS-based direction finding methods, where the directions are estimated by reconstructing the sparse signals in the spatial domain with the corresponding distribution assumptions of unknown parameters. Consequently, the SBL-based CS method, named CS-SBL, is developed in~\cite{shihao2008} to reconstruct the sparse signals. However, in the CS-SBL method, the discrete grids are adopted to formulate the CS-based system model, so the estimation performance is limited by the grid size. To further improve the estimation performance, dense grids can be adopted, but it will improve the computational complexity in the sparse reconstruction algorithm.

Additionally, the dense grids improve the correlation between the grids and decrease the performance of sparse reconstruction. {\hll The \emph{off-grid} CS-based methods were proposed~\cite{Xiaohuan2016,pengTSP} to overcome the grid problem in the CS-based model,} such as the off-grid sparse Bayesian inference (OGSBI) method proposed in~\cite{yang2013}. In the off-grid CS-based system model, the ground-truth directions are approximated by the Taylor expansion, so {\hll the performance of direction estimation}  can be improved in the off-grid methods when the same grids are adopted. Moreover, by solving the roots of a specific polynomial in an off-grid model, the Root-SBL method~\cite{jisheng2017} was also proposed to decrease the computational complexity of the SBL-based method. The grid evolution method was proposed in~\cite{qianli2018}  to refine the grids for the SBL-based method, and a dictionary learning algorithm is proposed in~\cite{Hojatollah2016}.

In a practical direction finding problem, the imperfection of the antenna array will decrease the estimation performance, so the mutual coupling effect between antennas cannot be ignored~\cite{Clerckx2007,Zhidong2012}. {\hll The direction finding methods were proposed in~\cite{Paolo2017,Jianyan2017,matthew2017} to decrease the mutual coupling effect.} However, in the existing sparse-based methods, the unknown mutual coupling effect is not considered, especially in a scenario with off-grid effect. 

In this paper, a symmetric Toeplitz matrix~\cite{Thomas2018,ce2017,liao2012} is used to describe the mutual coupling effect, and {\hll a} novel direction estimation method is proposed. With both the off-grid and the mutual coupling effect, {\hll the direction finding problem is investigated.} A novel system model is formulated to describe both the off-grid and the mutual coupling effect. Then, by exploiting the signal sparsity in the spatial domain, a novel direction finding method based on SBL with unknown mutual coupling effect, named DFSMC, is proposed. Additionally, with the distribution assumptions, we theoretically derive the estimation of all unknown parameters using the expectation-maximum (EM)-based  method in DFSMC, where the unknown parameters include the mutual coupling vector, the noise variance, the signals, the off-grid vector, et al. Finally, the proposed DFSMC method is compared with the  state-of-art methods in the direction finding performance. To summarize, we make the contributions as follows:
\begin{itemize}
    \item \textbf{The SBL-based system model with mutual coupling effect:} With considering both the off-grid and the unknown mutual coupling problems, a novel system model is formulated and transforms the direction finding problems into a sparse reconstruction problem.
    \item \textbf{The DFSM method for direction finding estimation:} With the distribution assumptions of all unknown parameters, a novel SBL-based direction finding method with unknown mutual coupling effect, named DFSMC, is proposed. DFSMC method estimates the directions via updating all the unknown parameters alternatively and achieves better estimation performance than the state-of-art methods.
    \item \textbf{The theoretical estimation expressions for all unknown parameters:} In the proposed DFSMC method, the EM method is adopted to estimate all the unknown parameters including the noise variance, the received signals, the mutual coupling vector, and the off-grid vectors, et al. With the distribution assumptions, we theoretically derive  the expressions for all the unknown parameters.  
\end{itemize}

The remainder of this paper is organized as follows. The system model for direction finding with unknown mutual coupling effect is formulated in Section~\ref{sec2}. The direction finding method based on SBL is given in Section~\ref{sec3}.  The simulation results are given in Section~\ref{sec4}. Finally, Section~\ref{sec5} concludes the paper.

\textit{Notations:}  $\boldsymbol{I}_N$ denotes an $N\times N$ identity matrix. 
$\mathcal{E}\left\{\cdot\right\}$ denotes the expectation operation.
$\mathcal{CN}\left(\boldsymbol{a},\boldsymbol{B}\right)$ denotes the complex Gaussian distribution with the mean being $\boldsymbol{a}$ and the covariance matrix being $\boldsymbol{B}$.  $\|\cdot\|_2$, $\otimes$,  $\operatorname{Tr}\left\{\cdot\right\}$,   $(\cdot)^*$, $(\cdot)^\text{T}$ and $(\cdot)^\text{H}$ denote the $\ell_2$ norm, the Kronecker product,  the trace of a matrix, the conjugate, the matrix transpose and the  Hermitian transpose, respectively. $\mathcal{R}\{a\}$ denotes the real part of complex value $a$. Additionally, for a vector $\boldsymbol{a}$, $\left[\boldsymbol{a}\right]_n$ denotes the $n$-th entry of $\boldsymbol{a}$, and $\operatorname{diag}\{\boldsymbol{a}\}$ denotes a diagonal matrix with the diagonal entries from $\boldsymbol{a}$. For a matrix $\boldsymbol{A}$, $\boldsymbol{A}_{:,n}$ denotes the $n$-th column of $\boldsymbol{A}$, and $\operatorname{diag}\{\boldsymbol{A}\}$ denotes a vector with the entries from the diagonal entries of $\boldsymbol{A}$.

\section{ULA System for Direction Finding}\label{sec2}

\begin{figure}
	\centering
	\includegraphics[width=3.5in]{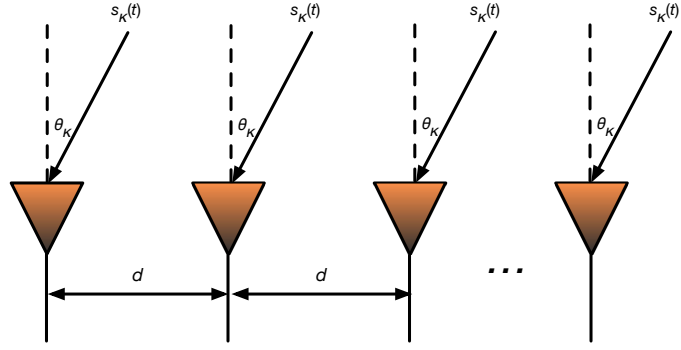}
	\caption{The ULA system for direction finding.}
	\label{system}
\end{figure}

We consider the direction finding problem in {\hll the uniform linear array (ULA) system}, where $N$ antennas are adopted and the inter-antenna element spacing is $d$. As shown in Fig.~\ref{system}, $K$ unknown signals with different directions ($\theta_k$, $K=0,1,\dots,K-1$) are received by the ULA. Thus, the received signals in the $N$ antennas can be expressed as 
\begin{align}\label{recv}
 	\boldsymbol{y}(t)&=\sum_{k=0}^{K-1}\boldsymbol{C}\boldsymbol{a}(\theta_k)s_k(t)+\boldsymbol{n}(t)  =\boldsymbol{CA}\boldsymbol{s}(t)+\boldsymbol{n}(t),
 \end{align}
 where the matrix $\boldsymbol{C}\in\mathbb{C}^{N\times N}$ denotes the mutual coupling matrix, and {\hll the signals are collected into a vector
 $\boldsymbol{s}(t)\triangleq \begin{bmatrix}
 	s_0(t),s_1(t),\dots,s_{K-1}(t)
 \end{bmatrix}^\text{T}$ with the $k$-th signal being  $s_k(t)$}. 
 Then, the received signals in the ULA antennas can be {\hll expressed} as
 $
 	\boldsymbol{y}(t)\triangleq\begin{bmatrix}
 	y_0(t),y_1(t),\dots,y_{N-1}(t)
 \end{bmatrix}^\text{T}$, 
and {\hll the zero-mean additive white Gaussian noise (AWGN) with the variance being $\sigma_\text{n}^2$} is 
$
	\boldsymbol{n}(t)\triangleq\begin{bmatrix}
 	n_0(t),n_1(t),\dots,n_{N-1}(t)
 \end{bmatrix}^\text{T}$. 
In this paper, we suppose that the noise variance $\sigma_\text{n}^2$ is unknown.  $\boldsymbol{A}\in\mathbb{C}^{N\times K}$ denotes the steering matrix for the $K$ signals, and can be expressed as
 \begin{align}
 	\boldsymbol{A}&\triangleq\begin{bmatrix}
 	\boldsymbol{a}(\theta_0),\boldsymbol{a}(\theta_1),\dots,\boldsymbol{a}(\theta_{K-1})
 \end{bmatrix},
 \end{align}
 where the steering vector for the $k$-th signal can be written as $
 	\boldsymbol{a}(\theta_k)\triangleq\begin{bmatrix}
 		a_0(\theta_k),a_1(\theta_k),\dots,a_{N-1}(\theta_k)
 	\end{bmatrix}^\text{T}$, $a_n(\theta_k)=e^{j2\pi\frac{nd}{\lambda}\sin\theta_k}$, and $\lambda $ denotes the wavelength. 

In this paper, we consider the direction finding problem with unknown mutual coupling effect between antennas, and the mutual coupling effect can be described usually by a symmetric Toeplitz matrix~\cite{ce2017}. As expressed in (\ref{recv}), the  mutual coupling matrix  can be represented as 
\begin{align}
	\boldsymbol{C}=  
	\begin{bmatrix}
	1 & c_{1} &\dots & c_{N-1}\\
	c_{1} & 1 & \dots & c_{N-2}\\
	\vdots & \vdots & \ddots & \vdots\\
	c_{N-1} & \dots & c_{1} & 1
	\end{bmatrix},\label{C}
\end{align}
where $c_n$ ($n=1,2,\dots, N-1$) denotes the mutual coupling coefficient between the $n_1$-th antenna and the $n_2$-th antenna, and $|n_1-n_2|=n$.

The signal model in (\ref{recv}) is a contiiuous domain model, and after the uniform sampling, a discrete model can be obtained in a matrix form as
\begin{align}\label{Y}
	\boldsymbol{Y}=\boldsymbol{CAS}+\boldsymbol{N},
\end{align}
where the sampling interval is $T_s$, the number of the samples is $M$. $\boldsymbol{Y}\in\mathbb{C}^{N\times M}$, $\boldsymbol{S}\in\mathbb{C}^{K\times M}$ and $\boldsymbol{N}\in\mathbb{C}^{N\times M}$ are expressed as $
	\boldsymbol{Y} =\begin{bmatrix}
		\boldsymbol{y}(0),\boldsymbol{y}(T_s),\dots, \boldsymbol{y}((M-1)T_{s})
	\end{bmatrix}$,  $
	\boldsymbol{S} =\begin{bmatrix}
		\boldsymbol{s}(0),\boldsymbol{s}(T_s),\dots,\boldsymbol{s}((M-1)T_{s})
	\end{bmatrix}$, $
	\boldsymbol{N} =\begin{bmatrix}
		\boldsymbol{n}(0),\boldsymbol{n}(T_s),\dots,\boldsymbol{n}((M-1)T_{s})
	\end{bmatrix}$. 
To simplify the notations, we define $\boldsymbol{y}_m\triangleq \boldsymbol{y}(mT_s)$, $\boldsymbol{n}_m\triangleq \boldsymbol{n}(mT_s)$ and $\boldsymbol{s}_m\triangleq \boldsymbol{s}(mT_s)$, so we have $\boldsymbol{Y}=\begin{bmatrix}
		\boldsymbol{y}_0,\boldsymbol{y}_1,\dots, \boldsymbol{y}_{M-1}
	\end{bmatrix}$, $
	\boldsymbol{S}=\begin{bmatrix}
		\boldsymbol{s}_0,\boldsymbol{s}_1,\dots,\boldsymbol{s}_{M-1}
	\end{bmatrix}$, $
	\boldsymbol{N}=\begin{bmatrix}
		\boldsymbol{n}_0,\boldsymbol{n}_1,\dots,\boldsymbol{n}_{M-1}
	\end{bmatrix}$.

However, the system model in (\ref{Y}) is hard to solve directly with the unknown mutual coupling matrix $\boldsymbol{C}$, so we try to express the matrix $\boldsymbol{C}$ in a vector form. The mutual coupling matrix in (\ref{C}) can be described alternatively by a vector $\boldsymbol{c}$ as $
	\boldsymbol{C}=\operatorname{Toeplitz}\{\boldsymbol{c}\}$,
where $\boldsymbol{c}\triangleq \begin{bmatrix}
	1,c_1,\dots,c_{N-1}
\end{bmatrix}^\text{T}$ is the first column of $\boldsymbol{C}$, and $\operatorname{Toeplitz}\{\cdot\}$ denotes the Toeplitz transformation. Therefore, after the simplification, the received signals during the $m$-th sampling interval in (\ref{Y}) can be rewritten as
\begin{align}\label{qq}
	\boldsymbol{y}_m&=\boldsymbol{CAs}_m+\boldsymbol{n}_m =\boldsymbol{Q}(\boldsymbol{s}_m\otimes \boldsymbol{c}) +\boldsymbol{n}_m, 
\end{align}
where the mutual coupling effect is expressed by a vector $\boldsymbol{c}$, and we use a matrix $\boldsymbol{Q}\in\mathbb{C}^{N\times KN}$ to rearrange the steering matrix $\boldsymbol{A}$.

According to the lemma in~\cite{Termos:2004ch,LIU2012517,ce2017}, the matrix $\boldsymbol{Q}$ in (\ref{qq}) can be obtained as 
\begin{align}
	\boldsymbol{Q}\triangleq\begin{bmatrix}
		\boldsymbol{Q}(\theta_0),\boldsymbol{Q}(\theta_1),\dots,\boldsymbol{Q}(\theta_{K-1})
	\end{bmatrix}.
\end{align}
The $k$-th sub-matrix $\boldsymbol{Q}(\theta_k)\in \mathbb{C}^{N\times N}$ is $
\boldsymbol{Q}(\theta_k)=\boldsymbol{Q}_{1}(\theta_k)+\boldsymbol{Q}_{2}(\theta_k)\label{Qa}$, 
where $\boldsymbol{Q}_{1}(\theta_k)$ and $\boldsymbol{Q}_{2}(\theta_k)$  respectively are
\begin{align}
	\boldsymbol{Q}_{1}(\theta_k)&\triangleq\begin{bmatrix}
a_0(\theta_k) &a_1(\theta_k) &\dots &a_{N-1}(\theta_k)\\
a_1(\theta_k) &a_2(\theta_k) &\dots &0\\
\vdots &\vdots &\ddots &\vdots\\
a_{N-2}(\theta_k) &a_{N-1}(\theta_k) &\dots &0\\
a_{N-1}(\theta_k) &0 &\dots &0
\end{bmatrix},\\
\boldsymbol{Q}_{2}(\theta_k)&\triangleq\begin{bmatrix}
0 &0 &\dots &0 &0\\
0 &a_0(\theta_k) &\dots &0&0\\
\vdots &\vdots &\ddots &\vdots & \vdots\\
0 &a_{N-3}(\theta_k) &\dots &a_{0}(\theta_k)&0\\
0 &a_{N-2}(\theta_k)&\dots &a_{1}(\theta_k)&a_{0}(\theta_k)
\end{bmatrix}.
\end{align}
Therefore, by collecting the $M$ samples into a matrix, the received signals in (\ref{Y}) can be finally rewritten as
\begin{align}\label{YY}
\boldsymbol{Y} =\boldsymbol{Q}(\boldsymbol{S}\otimes \boldsymbol{c}) +\boldsymbol{N}.
\end{align}

In this paper, we will propose a high resolution method to estimate the directions ($\theta_0$, $\theta_1$,\dots,$\theta_{K-1}$)   from the received signal matrix $\boldsymbol{Y}$, where the signal matrix $\boldsymbol{S}$, the mutual coupling vector $\boldsymbol{c}$ and the  noise variance  $\sigma^2_\text{n}$ are all unknown.


\section{Direction Finding Method Based on Sparse Bayesian Learning}\label{sec3}

In this section, we propose a novel SBL-based method to estimate the directions ({\hll named \textbf{D}irection \textbf{F}inding based on \textbf{S}BL with \textbf{M}utual \textbf{C}oupling effect, DFSMC}). The sparse model will be established first, and the DFSMC will be proposed with the distribution assumptions of unknown parameters.

\begin{figure}
	\centering
	\includegraphics[width=3in]{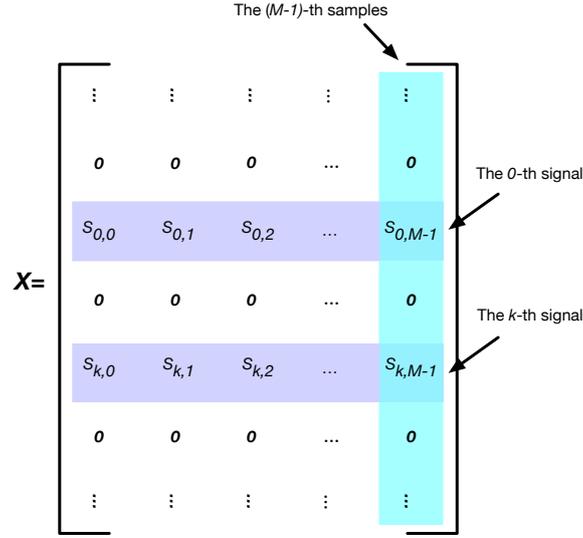}
	\caption{The structure of sparse matrix $\boldsymbol{X}$. }
	\label{Xmat}
\end{figure}

\subsection{Sparse-Based Signal Model}
Since the received signals are sparse in the spatial domain, we propose a sparse-based model to estimate the directions with unknown mutual coupling effect. In the sparse-based model, the dictionary matrix must be established first, {\hll so an over-complete dictionary matrix can be formulated by discretizing the signal direction  uniformly in the spatial domain}
\begin{align}
	\boldsymbol{D}\triangleq\begin{bmatrix}
		\boldsymbol{Q}(\zeta_0),\boldsymbol{Q}(\zeta_1),\dots,\boldsymbol{Q}(\zeta_{U-1})
	\end{bmatrix}\in\mathbb{C}^{N\times UN},
\end{align}
where $\zeta_u$ denotes the $u$-th discretized direction ($u=0,1,\dots,U-1$), $U$ denotes the number of discretized directions, the grid size is defined as  $\delta\triangleq |\zeta_{u+1}-\zeta_u|$, and we use a vector to contain all the discretized directions $\boldsymbol{\zeta}\triangleq \begin{bmatrix}
	\zeta_0,\zeta_1,\dots, \zeta_{U-1}
\end{bmatrix}$.

With the discretized directions and the system model in (\ref{YY}), a sparse-based \emph{on-grid} direction finding model can be expressed as
\begin{align}
	\boldsymbol{Y}=\boldsymbol{D}(\boldsymbol{X}\otimes \boldsymbol{c})+\boldsymbol{N},
\end{align}
where $\boldsymbol{X}$ is a sparse matrix
\begin{align}\label{SP}
	\boldsymbol{X}\triangleq \begin{bmatrix}
	\boldsymbol{x}_0,\boldsymbol{x}_1,\dots,\boldsymbol{x}_{M-1}
\end{bmatrix}\in\mathbb{C}^{U\times M}.
\end{align}
The structure of sparse matrix $\boldsymbol{X}$ is shown in Fig.~\ref{Xmat}, and 
the sparse vectors ($\boldsymbol{x}_0$, $\boldsymbol{x}_1$, $\dots$, $\boldsymbol{x}_{M-1}$) have the same support sets. When the direction of the $k$-th received signal $\theta_k$ is equal to the $u_k$-th discretized direction $\zeta_{u_k}$, we have $X_{u_k,m}=S_{k,m}$, so the $u$-th row and $m$-th column of $\boldsymbol{X} $ is 
\begin{align}
	X_{u,m}=\begin{cases}
		S_{k,m},&u=u_k\\
		0,&\text{otherwise}
	\end{cases}.
\end{align}

The sparse-based model in (\ref{SP}) assumes that the directions of received signals are exactly on the discretized grids. However, in the practical direction finding system, when the direction $\theta_k$ is not on the discretized grids, the direction $\theta_k$ can be represented by $\zeta_{u_k}$, which is a grid nearest to $\theta_k$. Thus, the corresponding matrix $\boldsymbol{Q}(\theta_k)$ in (\ref{Qa}) can be approximated by
\begin{align}\label{Q}
	\boldsymbol{Q}(\theta_k) \approx \boldsymbol{Q}(\zeta_{u_k}) +(\theta_k-\zeta_{u_k})\boldsymbol{\Omega}(\zeta_{u_k}),
\end{align}
where  the first-order derivative is defined as $	\boldsymbol{\Omega}(\zeta_{u_k}) \triangleq \left.\frac{\partial \boldsymbol{Q}(\zeta)}{\partial \zeta}\right|_{\zeta=\zeta_{u_k}}$. 
For example, as shown in Fig.~\ref{dis}, the direction of signal $s_k(t)$ is $\theta_k$, and the nearest grid is $\zeta_3$. Thus, the corresponding matrix $\boldsymbol{Q}(\theta_k)$ in (\ref{Q}) can be written as $
	\boldsymbol{Q}(\theta_k) \approx \boldsymbol{Q}(\zeta_{3}) +(\theta_k-\zeta_{3})\boldsymbol{\Omega}(\zeta_{3})$. 

\begin{figure}
	\centering
	\includegraphics[width=3.5in]{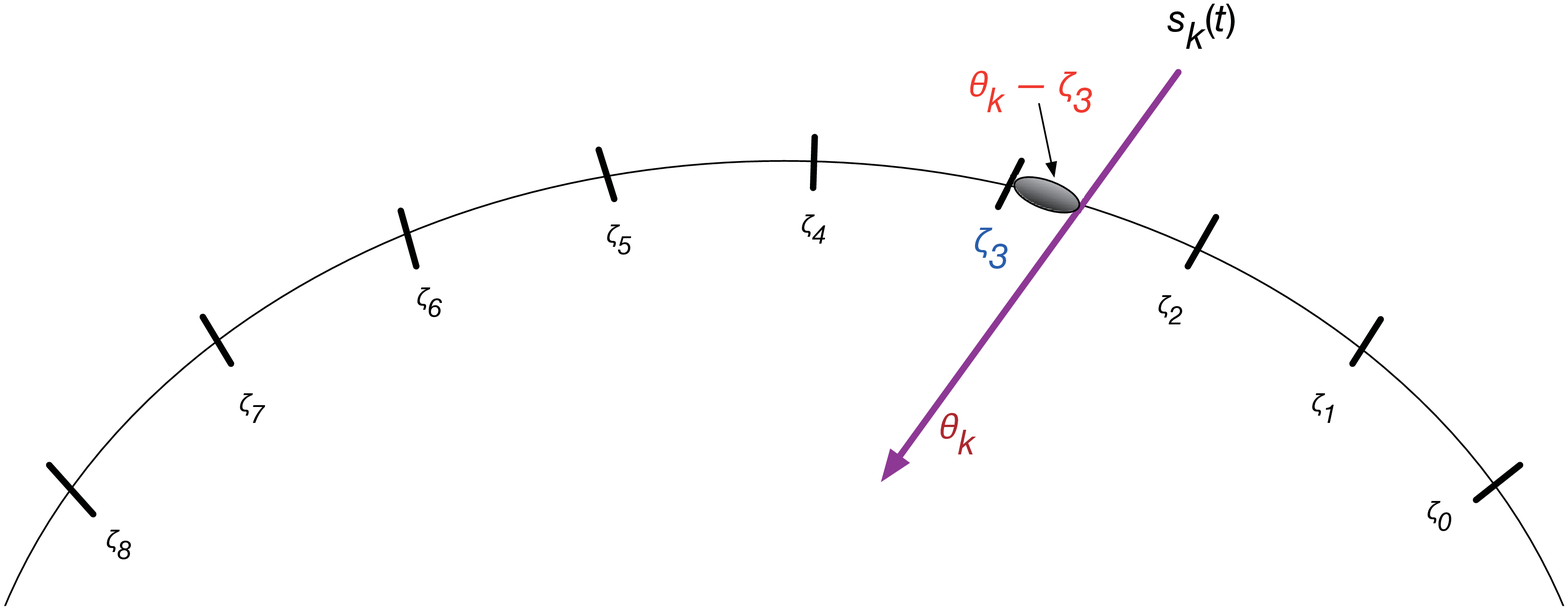}
	\caption{The off-grid approximation for direction. }
	\label{dis}
\end{figure}

Therefore, with the approximation in (\ref{Q}), the received signal in (\ref{YY}) can be approximated by a sparse-based \emph{off-grid} model 
\begin{align}\label{cs}
	\boldsymbol{Y}\approx \boldsymbol{\Psi}(\boldsymbol{\nu})(\boldsymbol{X}\otimes \boldsymbol{c})+\boldsymbol{N},
\end{align}
where $
	\boldsymbol{\Psi}(\boldsymbol{\nu})\triangleq \boldsymbol{D}+\boldsymbol{\Xi}\left(\operatorname{diag}\left\{
	\boldsymbol{\nu} \right\}\otimes \boldsymbol{I}_{N}\right)$, $
	\boldsymbol{\Xi}\triangleq\begin{bmatrix}
	\boldsymbol{\Xi}_0,\boldsymbol{\Xi}_1,\dots,\boldsymbol{\Xi}_{U-1}
\end{bmatrix}$, 
and the $u$-th submatrix of $\boldsymbol{\Xi}$ is denoted as $\boldsymbol{\Xi}_u\triangleq \boldsymbol{\Omega}(\zeta_u)$. Additionally, a vector $\boldsymbol{\nu}\in\mathbb{R}^{U\times 1}$ is used to represent the off-grid directions, and the $u$-th entry is
\begin{align}
	\nu_u=\begin{cases}
		\theta_k-\zeta_{u_k}, & u=u_k\\
		0, & \text{otherwise}
	\end{cases}.
\end{align}

Finally, an off-grid sparse-based model is formulated for the direction finding problem in (\ref{cs}).   We will estimate the directions by reconstructing the sparse matrix $\boldsymbol{X}$. The positions of non-zero entries in $\boldsymbol{X}$ indicate the directions of received signals. Simultaneously, the unknown parameters including the mutual coupling vector $\boldsymbol{c}$, the noise variance $\sigma^2_{\text{n}}$ and the off-grid vector $\boldsymbol{\nu}$ will also be estimated. 

\subsection{Distribution Assumptions}

\begin{figure*}
	\centering
	\includegraphics[width=5.5in]{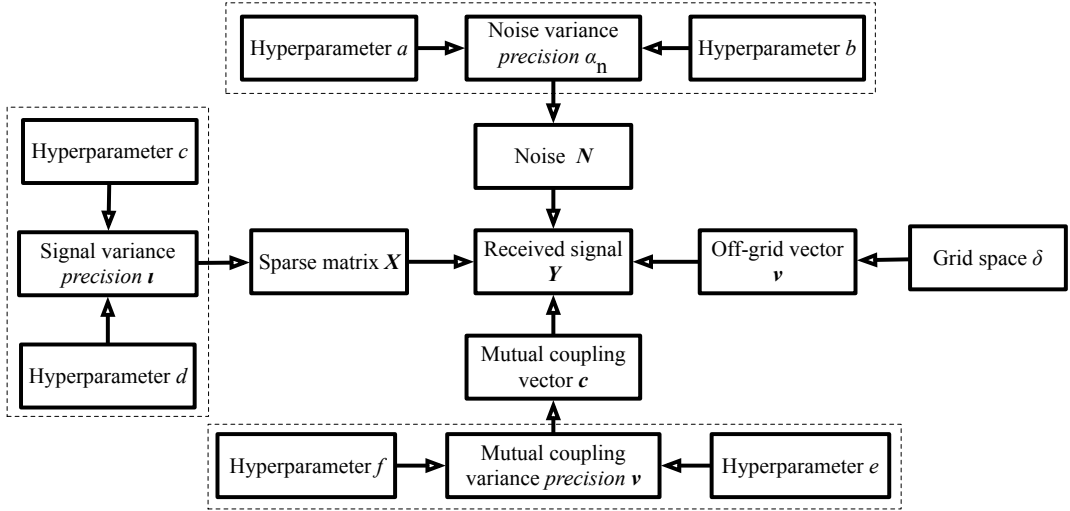}
	\caption{Graphical model of sparse bayesian learning for direction estimation. }
	\label{Bayesian}
\end{figure*}

In the proposed DFSMC method, the sparse Bayesian learning theory is adopted, and the method is established based on the distribution assumptions of all the unknown parameters. We assume that the unknown parameters follow the following distributions:
\begin{itemize}
	\item Noise $\boldsymbol{N}$: Gaussian distribution;
	\item The precision of noise variance $\alpha_{\text{n}}$: Gamma distribution;
	\item Sparse matrix $\boldsymbol{X}$: Gaussian distribution;
	\item The precision of signal variance $\boldsymbol{\iota}$: Gamma distribution;
	\item Mutual coupling vector $\boldsymbol{c}$: Gaussian distribution;
	\item The precision of mutual coupling variance $\boldsymbol{\vartheta}$: Gamma distribution;
	\item Off-grid vector $\boldsymbol{\nu}$: Uniform distribution.
\end{itemize}
The relationships between all the unknown parameters are given in Fig.~\ref{Bayesian}, and we will describe the distributions more clear in the following contents.

\subsubsection{The distribution of noise} 
When the received signals are independent between different samples, with the assumption of circular symmetric white Gaussian noise, the distribution of noise can be expressed as
\begin{align}
	p(\boldsymbol{N}|\sigma^2_\text{n})=\prod^{M-1}_{m=0}\mathcal{CN}(\boldsymbol{n}_m|\boldsymbol{0}_{N\times 1},\sigma^2_\text{n}\boldsymbol{I}_{N}),
\end{align}
where $\sigma^2_\text{n}$ denotes the noise variance, and the complex Gaussian distribution with the mean being $\boldsymbol{\mu}$ and the covariance matrix being $\boldsymbol{\Sigma}$ is expressed as
\begin{align}
	\mathcal{CN}(\boldsymbol{x}|\boldsymbol{\mu},\boldsymbol{\Sigma})=\frac{1}{\pi^N\det(\boldsymbol{\Sigma})}e^{-(\boldsymbol{x}-\boldsymbol{\mu})^\text{H}\boldsymbol{\Sigma}^{-1}(\boldsymbol{x}-\boldsymbol{\mu})}.
\end{align}
\subsubsection{The distribution of noise variance $\sigma_\text{n}^2$} 
In this paper, the noise variance is unknown.
Since the Gamma distribution is a conjugate prior of Gaussian distribution, the posterior distribution also follows a Gamma distribution. Therefore, using the Gamma distribution can simplify the following analysis.  With the unknown noise variance $\sigma^2_\text{n}$, we use a Gamma distribution to describe the precision of noise variance $\alpha_\text{n}\triangleq \sigma^{-2}_\text{n}$, and we have the following Gamma distribution
\begin{align}
	p(\alpha_\text{n})=g(\alpha_\text{n};a,b),
\end{align}
where $a$ and $b$ are the hyperparameters for $\alpha_\text{n}$, and $
	g(\alpha_\text{n};a,b) \triangleq \Gamma^{-1}(a)b^a\alpha_\text{n}^{a-1}e^{-b\alpha_\text{n}}$, $
	\Gamma(a) \triangleq \int^\infty_0x^{a-1}e^{-x}dx$.

\subsubsection{The distribution of sparse matrix}
With the independent received signals $\boldsymbol{S}$ among samples, we can assume that the sparse matrix $\boldsymbol{X}$ follows a zero-mean Gaussian distribution
\begin{align}
	p(\boldsymbol{X}|\boldsymbol{\Lambda}_\text{x})=\prod_{m=0}^{M-1}\mathcal{CN}(\boldsymbol{x}_m|\boldsymbol{0}_{U\times 1}, \boldsymbol{\Lambda}_\text{x}),
\end{align}
where the covariance matrix $\boldsymbol{\Lambda}_\text{x}\in\mathbb{R}^{U\times U}$ is a diagonal matrix with the $u$-th diagonal entry being $\sigma^2_{\text{x},u}$. 
\subsubsection{The distribution of signal variance}
Similarity, with the unknown signal variance $\boldsymbol{\Lambda}_{\text{x}}$, we define the precision vector $
	\boldsymbol{\iota}\triangleq \begin{bmatrix}
	\iota_0,\iota_1,\dots,\iota_{U-1}
\end{bmatrix}^\text{T}$, 
where $\iota_u\triangleq \sigma^{-2}_{\text{x},u}$, so $\boldsymbol{\iota}$ can be expressed by a Gamma prior
\begin{align}
	p(\boldsymbol{\iota};c,d)=\prod^{U-1}_{u=0}g(\iota_u;c,d),
\end{align} 
where $c$ and $d$ are the hyperparmaters for $\boldsymbol{\iota}$.
\subsubsection{The distribution of mutual coupling vector}
With the unknown mutual coupling vector, when the  mutual coupling coefficients are independent between antennas, the distribution of mutual coupling vector $\boldsymbol{c}$ can be expressed as a Gaussian distribution
\begin{align}
	p(\boldsymbol{c}|\boldsymbol{\Lambda}_\text{c})&=\prod^{N-1}_{n=0}\mathcal{CN}(c_{n}|0,\boldsymbol{\Lambda}_\text{c}),\end{align}
where the covariance matrix $\boldsymbol{\Lambda}_\text{c}\in\mathbb{R}^{N\times N}$ is a diagonal matrix with the $n$-th diagonal entry being $\sigma^2_{\text{c},n}$. 
\subsubsection{The distribution of mutual coupling variance}
By defining the precisions $\boldsymbol{\vartheta}\triangleq\begin{bmatrix}
	\vartheta_{0},\vartheta_{1},\dots,\vartheta_{N-1}
\end{bmatrix}^\text{T}$ ($\vartheta_{n}\triangleq\sigma^{-2}_{\text{c}, n}$), we  use a Gamma distribution to describe the distribution of $
\boldsymbol{\vartheta}$
\begin{align}
p(\boldsymbol{\vartheta}; e, f)&=\prod^{N-1}_{n=0}g(\vartheta_{n};e,f),\end{align}
where both $e$ and $f$ are the hyperparameters of $\boldsymbol{\vartheta}$. 

\subsubsection{The distribution of off-grid vector}
We can assume that the off-grid vector $\boldsymbol{\nu}$ follows a uniform  distribution, and the distribution of $\nu_u$ can be expressed as
\begin{align}
	p(\nu_u;\delta)=\mathcal{U}_{\nu_u}\left(\left[-\frac{1}{2}\delta,\frac{1}{2}\delta\right]\right),
\end{align}
where the uniform distribution is defined as
\begin{align}
	\mathcal{U}_{x}\left(\left[a,b\right]\right)\triangleq\begin{cases}
		\frac{1}{b-a}, &a\leq x\leq b\\
		0,&\text{otherwise}
	\end{cases}.
\end{align}


\subsection{DFSMC method}
With the distribution assumptions of unknown parameters, a novel direction finding method based on the SBL is proposed with the unknown mutual coupling effect, named DFSMC. In the SBL-based method, the posterior probabilities for all the unknown parameters are theoretically derived. 

To estimate the directions, we formulate the following problem to maximize the posterior probability 
\begin{align}\label{wp}
	\hat{\wp}=\arg\max_{\wp}p(\wp|\boldsymbol{Y}),
\end{align}
where a set $\wp\triangleq  \left\{ \boldsymbol{X},\boldsymbol{\nu},\boldsymbol{c}, \sigma^2_\text{n},\boldsymbol{\iota},\boldsymbol{\vartheta}\right\}$ is used to contain all the unknown parameters. 
However, the problem (\ref{wp}) is too complex and cannot be solved directly. The expectation maximum (EM)-based method is used to realize the proposed DFSMC method. Additionally, with the received signal $\boldsymbol{Y}$, the joint distribution with unknown parameters can be expressed as 
\begin{align}\label{diss}
	p(\boldsymbol{Y},\wp) &=p(\boldsymbol{Y}|\wp)p(\boldsymbol{X}|\boldsymbol{\iota}) p(\boldsymbol{c}|\boldsymbol{\vartheta})p(\alpha_\text{n}) p(\boldsymbol{\iota})p(\boldsymbol{\vartheta})p(\boldsymbol{\nu}).
\end{align}
The details to estimate all unknown parameters are given as follows.

\subsubsection{The sparse matrix}
Given the received signal $\boldsymbol{Y}$ and the parameters  $(\wp\backslash \boldsymbol{X})$ excepting $\boldsymbol{X}$, the the posterior of $\boldsymbol{X}$ can be expressed as 
\begin{align}
	p(\boldsymbol{X}|\boldsymbol{Y},\boldsymbol{\nu},\boldsymbol{c}, \alpha_\text{n},\boldsymbol{\iota},\boldsymbol{\vartheta}) & = \frac{p(\boldsymbol{Y}|\wp)p(\boldsymbol{X}|\boldsymbol{\iota})}{p(\boldsymbol{Y}|\boldsymbol{\nu},\boldsymbol{c}, \alpha_\text{n},\boldsymbol{\iota},\boldsymbol{\vartheta})} \propto p(\boldsymbol{Y}|\wp)p(\boldsymbol{X}|\boldsymbol{\iota})\label{X},
\end{align}
where both $p(\boldsymbol{Y}|\wp)$ and $p(\boldsymbol{X}|\boldsymbol{\iota})$ follow Gaussian distributions, and can be calculated as
\begin{align}
	p(\boldsymbol{Y}|\wp) 	& = \prod^{M-1}_{m=0}\frac{\alpha_\text{n}^{N}}{\pi^{N}}e^{-\alpha_\text{n}\|\boldsymbol{y}_m-\boldsymbol{\Psi}(\boldsymbol{\nu})(\boldsymbol{x}_m\otimes \boldsymbol{c})\|^2_2},\label{YYY}\\
p(\boldsymbol{X}|\boldsymbol{\iota}) 
& = \prod_{m=0}^{M-1} \left(\prod^{U-1}_{u=0}\iota_u\right) \frac{1}{\pi^U}e^{-\boldsymbol{x}^\text{H}_m\operatorname{diag}\{\boldsymbol{\iota}\}\boldsymbol{x}_m}. \label{XXX}
\end{align}
Therefore, the posterior of $\boldsymbol{X}$ is also a Gaussian function
\begin{align}\label{XX}
	&p(\boldsymbol{X}|\boldsymbol{Y},\boldsymbol{\nu},\boldsymbol{c}, \alpha_\text{n},\boldsymbol{\iota},\boldsymbol{\vartheta}) = \prod^{M-1}_{m=0}\mathcal{CN}(\boldsymbol{x}_m|\boldsymbol{\mu}_m,\boldsymbol{\Sigma}_\text{X}),
\end{align}
where the mean $\boldsymbol{\mu}_m$ and covariance matrix $\boldsymbol{\Sigma}_\text{X}$ are obtained from (\ref{YYY}) and (\ref{XXX}) as
\begin{align}
	\boldsymbol{\mu}_m&=\alpha_\text{n}^\text{H} \boldsymbol{\Sigma}^\text{H}_\text{X}
	\mathfrak{T}^H(\boldsymbol{\nu},\boldsymbol{c}) \boldsymbol{y}_m,\label{x}\\
	\boldsymbol{\Sigma}_\text{X}&=\left[\alpha_\text{n} 
	\mathfrak{T}^H(\boldsymbol{\nu},\boldsymbol{c})\mathfrak{T}(\boldsymbol{\nu},\boldsymbol{c})
	+ \operatorname{diag}\{\boldsymbol{\iota}\}\right]^{-1}\label{sigma},
\end{align}
and we define the following function
\begin{align}\label{TT}
	\mathfrak{T}(\boldsymbol{\nu},\boldsymbol{c})\triangleq \boldsymbol{\Psi}(\boldsymbol{\nu})(\boldsymbol{I}_{U}
	\otimes \boldsymbol{c}).
\end{align}
Additionally, to simplify the notations, the $u$-th entry of $\boldsymbol{\mu}_m$ is denoted as $\mu_{u,m}$, and we can collect all the mean $\boldsymbol{\mu}_m$ as a matrix $
	\boldsymbol{\mu}\triangleq\begin{bmatrix}
	\boldsymbol{\mu}_0,\boldsymbol{\mu}_1,\dots,\boldsymbol{\mu}_{M-1}
\end{bmatrix}$.

To estimate the other unknown parameters $\wp\backslash \boldsymbol{X}$, with (\ref{diss}), we can formulate the following likelihood function 
\begin{align}
	 \mathcal{L}(\boldsymbol{\nu},\boldsymbol{c}, \alpha_\text{n},\boldsymbol{\iota},\boldsymbol{\vartheta}) = \mathcal{E}\big\{&\ln  p(\boldsymbol{Y}|\wp)p(\boldsymbol{X}|\boldsymbol{\iota})  p(\boldsymbol{c}|\boldsymbol{\vartheta})p( \alpha_\text{n}) p(\boldsymbol{\iota})p(\boldsymbol{\vartheta})p(\boldsymbol{\nu})\big\}.  
\end{align}
where we just use $\mathcal{E}\{\cdot\}$ to represent $\mathcal{E}_{\boldsymbol{X}|\boldsymbol{Y},\boldsymbol{\nu}, \alpha_\text{n},\boldsymbol{\iota},\boldsymbol{\vartheta}}\{\cdot\}$. 
Thus, the EM-based method can be used to estimate $\wp\backslash \boldsymbol{X}$, and the details are given in the following  contents.  Additionally, the derivatives for the complex vector and matrix are given as the following lemma.
\begin{lemma}\label{lemma}
With both the complex vectors ($\boldsymbol{u}\in\mathbb{C}^{P\times 1}$, $\boldsymbol{v}\in\mathbb{C}^{P\times 1}$) and the complex matrix $\boldsymbol{A}\in\mathbb{C}^{M\times P}$ being the function of a complex vector $\boldsymbol{x}\in\mathbb{C}^{N\times 1}$, the following derivations can be obtained 
\begin{align}
&\frac{\partial\boldsymbol{u}^\text{H}\boldsymbol{v}  }{\partial \boldsymbol{x}} = \boldsymbol{v}^\text{T} \frac{\partial (\boldsymbol{u}^*)}{\partial \boldsymbol{x}} 
+\boldsymbol{u}^\text{H}\frac{\partial \boldsymbol{v}}{\partial \boldsymbol{x}},\\
&\frac{\partial \boldsymbol{Au}}{\partial \boldsymbol{x}}= \begin{bmatrix}
\frac{\partial \boldsymbol{A}}{\partial x_0}\boldsymbol{u} + \boldsymbol{A}\frac{\partial \boldsymbol{u}}{\partial x_0},
\ldots, \frac{\partial \boldsymbol{A}}{\partial x_n}\boldsymbol{u} + \boldsymbol{A}\frac{\partial \boldsymbol{u}}{\partial x_n} ,\ldots
\end{bmatrix}.
\end{align}
\end{lemma}

\begin{proof}
See: Appendix~\ref{derivation}.	
\end{proof}

\subsubsection{The mutual coupling vector}
Ignoring terms independent thereof in $\mathcal{L}(\boldsymbol{\nu},\boldsymbol{c}, \alpha_\text{n},\boldsymbol{\iota},\boldsymbol{\vartheta})$, we can obtain the following likelihood function for the mutual coupling vector $\boldsymbol{c}$
\begin{align}
	\mathcal{L}&(\boldsymbol{\boldsymbol{c}})  = \mathcal{E}\left\{\ln  p(\boldsymbol{Y}|\boldsymbol{X},\boldsymbol{\nu},\boldsymbol{c}, \alpha_\text{n}) p(\boldsymbol{c}|\boldsymbol{\vartheta})\right\}\notag\\
	&= \mathcal{E}\left\{\ln   \prod^{M-1}_{m=0}\mathcal{CN}(\boldsymbol{y}_m|\boldsymbol{\Psi}(\boldsymbol{\nu})(\boldsymbol{x}_m\otimes \boldsymbol{c}),\alpha^{-1}_\text{n}\boldsymbol{I}_{N})\right\} + \ln
	 \prod^{N-1}_{n=0}\mathcal{CN}(c_{n}|0,\vartheta^{-1}_{n}) \\
	 	 &\propto  -\alpha_\text{n} M \mathcal{G}_{1} (\boldsymbol{c},\boldsymbol{\nu})  -\sum^{M-1}_{m=0} \alpha_\text{n} \mathcal{G}_{2,m}(\boldsymbol{c},\boldsymbol{\nu}) - \mathcal{G}_3(\boldsymbol{c}),\notag
	 \end{align}	  
where we have
\begin{align}
	 \mathcal{E}\left\{\ln  p(\boldsymbol{Y}|\boldsymbol{X},\boldsymbol{\nu},\boldsymbol{c}, \alpha_\text{n}) \right\}   = &MN\ln  \frac{\alpha_{\text{n}}}{\pi} 
	 -\alpha_\text{n} M \mathcal{G}_{1} (\boldsymbol{c},\boldsymbol{\nu}) -\sum^{M-1}_{m=0} \alpha_\text{n} \mathcal{G}_{2,m}(\boldsymbol{c},\boldsymbol{\nu}),  
\end{align}	 
and $
	\mathcal{G}_{1} (\boldsymbol{c},\boldsymbol{\nu}) \triangleq
	 \operatorname{Tr}\left\{ \mathfrak{T}^H(\boldsymbol{\nu},\boldsymbol{c})\mathfrak{T}(\boldsymbol{\nu},\boldsymbol{c}) \boldsymbol{\Sigma}_\text{X}\right\}$, $\mathcal{G}_{2,m}(\boldsymbol{c},\boldsymbol{\nu})  \triangleq \|\boldsymbol{y}_m-\boldsymbol{\Psi}(\boldsymbol{\nu})(\boldsymbol{\mu}_m\otimes \boldsymbol{c})\|^2_2$, $
	  \mathcal{G}_3(\boldsymbol{c}) \triangleq \sum^{N-1}_{n=0}  \vartheta_{n} |c_{n}|^2$. 

To estimate the mutual coupling vector $\boldsymbol{c}$, we can maximize the likelihood function $\mathcal{L}(\boldsymbol{\boldsymbol{c}})$, and we have
\begin{align}
	\hat{\boldsymbol{c}}=\arg\max_{\boldsymbol{c}} \mathcal{L}(\boldsymbol{c}).
\end{align}
Therefore,  by setting $\frac{\partial \mathcal{L}(\boldsymbol{c})}{\partial \boldsymbol{c}}=\boldsymbol{0}$,  the mutual coupling vector can be obtained. We can calculate
\begin{align}\label{abcd}
	\frac{\partial \mathcal{L}(\boldsymbol{c})}{\partial \boldsymbol{c}}&= 
	 -\alpha_\text{n} M \frac{\partial \mathcal{G}_{1} (\boldsymbol{c},\boldsymbol{\nu})}{\partial \boldsymbol{c}}  -\sum^{M-1}_{m=0} \alpha_\text{n} \frac{\partial \mathcal{G}_{2,m}(\boldsymbol{c},\boldsymbol{\nu})}{\partial \boldsymbol{c}} - \frac{\partial \mathcal{G}_3(\boldsymbol{c})}{\partial \boldsymbol{c}}.
\end{align}

In (\ref{abcd}), $\frac{\partial \mathcal{G}_{1} (\boldsymbol{c},\boldsymbol{\nu}) }{\partial \boldsymbol{c}}$, $\frac{\partial \mathcal{G}_{2,m} (\boldsymbol{c},\boldsymbol{\nu}) }{\partial \boldsymbol{c}}$ and $\frac{\partial \mathcal{G}_{3} (\boldsymbol{c}) }{\partial \boldsymbol{c}}$ can be calculated as follows.
\begin{itemize}
	\item \emph{For $\frac{\partial \mathcal{G}_{1} (\boldsymbol{c},\boldsymbol{\nu}) }{\partial \boldsymbol{c}}$:} With the derivations of complex vector and matrix in Appendix~\ref{derivation}, $\frac{\partial \mathcal{G}_{1} (\boldsymbol{c},\boldsymbol{\nu}) }{\partial \boldsymbol{c}}$ is a row vector, and the $n$-th entry can be calculated as 
\begin{align}\label{60}
	&\left[\frac{\partial \mathcal{G}_{1} (\boldsymbol{c},\boldsymbol{\nu}) }{\partial \boldsymbol{c}}\right]_n = \operatorname{Tr}\left\{
\frac{\partial \mathfrak{T}^H(\boldsymbol{\nu},\boldsymbol{c})\mathfrak{T}(\boldsymbol{\nu},\boldsymbol{c}) \boldsymbol{\Sigma}_\text{X} }{\partial c_{n}}\right\}. 
	 \end{align}
	 Additionally, we can calculate 	
	 \begin{align}
	 & \frac{\partial 
	 \mathfrak{T}^H(\boldsymbol{\nu},\boldsymbol{c})\mathfrak{T}(\boldsymbol{\nu},\boldsymbol{c})\boldsymbol{\Sigma}_\text{X} }{\partial c_{n}} = \frac{\partial 
	 (\boldsymbol{I}_U\otimes \boldsymbol{c})^\text{H}
	 }{\partial c_{n}}\boldsymbol{\Psi}^\text{H}(\boldsymbol{\nu})
	 \mathfrak{T}(\boldsymbol{\nu},\boldsymbol{c})\boldsymbol{\Sigma}_\text{X}] +\mathfrak{T}^H(\boldsymbol{\nu},\boldsymbol{c})
	 \boldsymbol{\Psi}(\boldsymbol{\nu})\frac{\partial 
	 (\boldsymbol{I}_U\otimes \boldsymbol{c}) }{\partial c_{n}}\boldsymbol{\Sigma}_\text{X}\notag\\ 
	 & = \mathfrak{T}^H(\boldsymbol{\nu},\boldsymbol{c})
	 \boldsymbol{\Psi}(\boldsymbol{\nu})\left(\boldsymbol{I}_U\otimes   \frac{\partial
	\boldsymbol{c} }{\partial c_{n}}\right)\boldsymbol{\Sigma}_\text{X}  = \mathfrak{T}^H(\boldsymbol{\nu},\boldsymbol{c})
	 \mathfrak{T}(\boldsymbol{\nu},\boldsymbol{e}^N_n) \boldsymbol{\Sigma}_\text{X},
\end{align}	 
where $\boldsymbol{e}^N_n$ is a $N\times 1$ vector with the $n$-th entry being $1$ and  other entries being $0$. Therefore, the the $n$-th entry in (\ref{60}) can be simplified as
\begin{align}
	\left[\frac{\partial \mathcal{G}_{1} (\boldsymbol{c},\boldsymbol{\nu}) }{\partial \boldsymbol{c}}\right]_n&=\boldsymbol{c}^\text{H}\left(\sum_{p=0}^{U-1}\sum_{k=0}^{U-1}
	 \boldsymbol{\Psi}^\text{H}_p(\boldsymbol{\nu})\boldsymbol{\Psi}_k(\boldsymbol{\nu})
	 \Sigma_{\text{X},k,p}
	 \right) \boldsymbol{e}^N_n,
\end{align}
and we finally have the derivation of $\mathcal{G}_{1} (\boldsymbol{c},\boldsymbol{\nu})$ as
\begin{align}
	\frac{\partial \mathcal{G}_{1} (\boldsymbol{c},\boldsymbol{\nu}) }{\partial \boldsymbol{c}} &=\boldsymbol{c}^\text{H}\left(\sum_{p=0}^{U-1}\sum_{k=0}^{U-1}
	 \boldsymbol{\Psi}^\text{H}_p(\boldsymbol{\nu})\boldsymbol{\Psi}_k(\boldsymbol{\nu})
	 \Sigma_{\text{X},k,p}
	 \right).\notag
\end{align}

\item $\frac{\partial \mathcal{G}_{2,m} (\boldsymbol{c},\boldsymbol{\nu}) }{\partial \boldsymbol{c}}$  can be simplified as
\begin{align}
	\frac{\partial \mathcal{G}_{2,m} (\boldsymbol{c},\boldsymbol{\nu}) }{\partial \boldsymbol{c}}&=-[\boldsymbol{y}_m-\boldsymbol{\Psi}(\boldsymbol{\nu})(\boldsymbol{\mu}_m\otimes \boldsymbol{c})]^\text{H} \boldsymbol{\Psi}(\boldsymbol{\nu}) \frac{\partial \boldsymbol{\mu}_m\otimes \boldsymbol{c}}{\partial \boldsymbol{c}}\notag \\
		&=-[\boldsymbol{y}_m-\boldsymbol{\Psi}(\boldsymbol{\nu})(\boldsymbol{\mu}_m\otimes \boldsymbol{c})]^\text{H} \boldsymbol{\Psi}(\boldsymbol{\nu})  (\boldsymbol{\mu}_m\otimes \boldsymbol{I}_N). 
\end{align}

\item $\frac{\partial \mathcal{G}_{3} (\boldsymbol{c}) }{\partial \boldsymbol{c}}$ can be simplified as $
	\frac{\partial \mathcal{G}_{3} (\boldsymbol{c})}{\partial \boldsymbol{c}}=\boldsymbol{c}^{H}\operatorname{diag}\{\boldsymbol{\vartheta}\}$. 

\end{itemize}

Therefore, with (\ref{abcd}), the mutual coupling vector can be finally estimated as
\begin{align}\label{ctE}
	\hat{\boldsymbol{c}}=\boldsymbol{H}^{-1}\boldsymbol{z},
\end{align}
where
\begin{align}
	\boldsymbol{H}& =  \sum^{M-1}_{m=0} \alpha_\text{n}  
		\mathfrak{P}^H(\boldsymbol{\nu},\boldsymbol{\mu}_m)\mathfrak{P}(\boldsymbol{\nu},\boldsymbol{\mu}_m)  + \alpha_\text{n} M 
	 \left(\sum_{p=0}^{U-1}\sum_{k=0}^{U-1}
	 \boldsymbol{\Psi}^\text{H}_p(\boldsymbol{\nu})\boldsymbol{\Psi}_k(\boldsymbol{\nu})
	 \Sigma_{\text{X},k,p}
	 \right)^\text{H} +\operatorname{diag}\{\boldsymbol{\vartheta}\}, \\
	 \boldsymbol{z}&=  \sum^{M-1}_{m=0} \alpha_\text{n} \mathfrak{P}^H(\boldsymbol{\nu},\boldsymbol{\mu}_m)		\boldsymbol{y}_m,
		\end{align}
and we define $
		\mathfrak{P}(\boldsymbol{\nu},\boldsymbol{\mu}_m)   \triangleq \boldsymbol{\Psi}(\boldsymbol{\nu})(\boldsymbol{\mu}_m\otimes  \boldsymbol{I}_N)$. 
		
\subsubsection{For the precision of signal variance}		
Ignoring terms independent thereof in $\mathcal{L}(\boldsymbol{\nu},\boldsymbol{c}, \alpha_\text{n},\boldsymbol{\iota},\boldsymbol{\vartheta})$, we can obtain the likelihood function of $\boldsymbol{\iota}$ as 
\begin{align}
	\mathcal{L}(\boldsymbol{\iota})  = &\mathcal{E} \left\{\ln  p(\boldsymbol{X}|\boldsymbol{\iota}) p(\boldsymbol{\iota}) \right\} = \mathcal{E}\left\{\ln  \prod_{m=0}^{M-1}\mathcal{CN}(\boldsymbol{x}_m|\boldsymbol{0}_{U\times 1}, \boldsymbol{\Lambda}_\text{x})\right\}  +\ln
	 \prod^{U-1}_{u=0}g(\iota_u;c,d). 
\end{align}
Then, the precision of signal variance can be estimated by $
	\hat{\boldsymbol{\iota}}=\arg\max_{\boldsymbol{\iota}}\mathcal{L}(\boldsymbol{\iota})$. 

By setting $\frac{\partial \mathcal{L}(\boldsymbol{\iota})}{\partial \boldsymbol{\iota}}=0$, the $u$-th entry of $\boldsymbol{\iota}$ can be obtained as 
\begin{align}\label{beta1}
	\hat{\iota}_u =\frac{M+c-1}{d + M\Sigma_{\text{X},u,u} + \sum_{m=0}^{M-1} |\mu_{u,m}|^2}. 
\end{align}
{\hll In the iterative algorithm, (\ref{beta1}) can be rewritten as
\begin{align}\label{beta}
\hat{\iota}^{i+1}_u \approx\frac{M-1-\hat{\iota}^{i}_u \sum_{m=0}^{M-1} |\mu_{u,m}|^2}{d + M\Sigma_{\text{X},u,u}},
\end{align}	
where $\hat{\iota}^{i+1}_u$ and $\hat{\iota}^{i}_u$ are the esitmated results at the $(i+1)$-th and the $i$-th iterations, respectively.
}

\subsubsection{For $\alpha_\text{n}$}
Ignoring terms independent thereof in $\mathcal{L}(\boldsymbol{\nu},\boldsymbol{c}, \alpha_\text{n},\boldsymbol{\iota},\boldsymbol{\vartheta})$, we can obtain the likelihood function
\begin{align}
	 \mathcal{L}(\alpha_\text{n})&= \mathcal{E}\left\{\ln  p(\boldsymbol{Y}|\boldsymbol{X},\boldsymbol{\nu},\boldsymbol{c}, \alpha_\text{n})p(\alpha_\text{n}) \right\}\notag\\
	& = \mathcal{E}\left\{\ln  
	\prod_{m=0}^{M-1}\mathcal{CN}\left(\boldsymbol{y}_m|\boldsymbol{\Psi}(\boldsymbol{\nu})(\boldsymbol{x}_m\otimes \boldsymbol{c}), \sigma^2_\text{n}\boldsymbol{I}\right)
	 \right\}  +\ln
	g(\alpha_\text{n};a,b).
\end{align}
The precision of noise variance can be estimated by
\begin{align}\label{alphan}
	\hat{\alpha}_{\text{n}}=\arg\max_{\alpha_{\text{n}}}\mathcal{L}(\alpha_{\text{n}}).
\end{align}

By setting $\frac{\partial \mathcal{L}(\alpha_\text{n})}{\partial \alpha_\text{n}}=0$, we can obtain
\begin{align}\label{n1}
\hat{\alpha}_\text{n}&=\frac{MN+a-1}{M\mathcal{G}_{1} (\boldsymbol{c},\boldsymbol{\nu}) + \sum^{M-1}_{m=0}\mathcal{G}_{2,m}(\boldsymbol{c},\boldsymbol{\nu})  +b}.
\end{align}
 {\hll In the iterative algorithm, (\ref{n1}) can be rewritten as
 	\begin{align}\label{n}
 	\hat{\alpha}^{i+1}_\text{n}\approx\frac{MN-1- \hat{\alpha}^{i}_\text{n}\sum^{M-1}_{m=0}\mathcal{G}_{2,m}(\boldsymbol{c},\boldsymbol{\nu})}{M\mathcal{G}_{1} (\boldsymbol{c},\boldsymbol{\nu})   +b}.,
 	\end{align}	
 	where $\hat{\alpha}^{i+1}_\text{n}$ and $\hat{\alpha}^{i}_\text{n}$ are the esitmated results at the $(i+1)$-th and the $i$-th iterations, respectively.
 }

%

\subsubsection{For the precision of mutual coupling variance}
Ignoring terms independent thereof in $\mathcal{L}(\boldsymbol{\nu},\boldsymbol{c}, \alpha_\text{n},\boldsymbol{\iota},\boldsymbol{\vartheta})$, we can obtain the likelihood function 
\begin{align}
	\mathcal{L}(\boldsymbol{\vartheta}) &= \mathcal{E}\left\{\ln p(\boldsymbol{c}|\boldsymbol{\vartheta}) p(\boldsymbol{\vartheta}) \right\} = \mathcal{E}\left\{\ln  
	\prod^{N-1}_{n=0}\mathcal{CN}(c_{n}|0,\vartheta_{n}^{-1})\right\} +\ln
	\prod^{N-1}_{n=0}g(\vartheta_{n};e,f). 
\end{align}
The precision of mutual coupling variance can be estimated by $
	\hat{\boldsymbol{\vartheta}}=\arg\max_{\boldsymbol{\vartheta}}\mathcal{L}(\boldsymbol{\vartheta})$. 

{\hll
By setting $\frac{\partial \mathcal{L}(\boldsymbol{\vartheta})}{\partial \boldsymbol{\vartheta}}=\boldsymbol{0}$, we can obtain the $n$-th entry of $\boldsymbol{\vartheta}$ as 
\begin{align}\label{varT}
	 \hat{\vartheta}_{n} \approx \frac{1}{f+c_{n}^\text{H}c_{n}}.
\end{align}
}			

\subsubsection{For the off-grid vector}

Ignoring terms independent thereof in $\mathcal{L}(\boldsymbol{\nu},\boldsymbol{c}, \alpha_\text{n},\boldsymbol{\iota},\boldsymbol{\vartheta})$, we can obtain the likelihood function
\begin{align} 
	\mathcal{L}(\boldsymbol{\nu}) & = \mathcal{E} \left\{\ln  p(\boldsymbol{Y}|\boldsymbol{X},\boldsymbol{\nu},\boldsymbol{c}, \alpha_n) p(\boldsymbol{\nu})\right\} \propto - M\mathcal{G}_{1} (\boldsymbol{c},\boldsymbol{\nu})
	 -\sum^{M-1}_{m=0} \mathcal{G}_{2,m}(\boldsymbol{c},\boldsymbol{\nu}).  
	 \end{align}
	 
The off-grid vector can be estimated by
\begin{align}
	\hat{\boldsymbol{\nu}}=\arg\max_{\boldsymbol{\nu}}\mathcal{L}(\boldsymbol{\nu}).
\end{align}	 

Then,
$\frac{\partial \mathcal{G}_{1} (\boldsymbol{c},\boldsymbol{\nu})}{\partial \boldsymbol{\nu}} \in \mathbb{R}^{1\times U}$ is a row vector, and the $u$-th entry is
\begin{align}
&\left[\frac{\partial \mathcal{G}_{1} (\boldsymbol{c},\boldsymbol{\nu})}{\partial \boldsymbol{\nu}}\right]_u  = \operatorname{Tr}\left\{\frac{\partial 
\mathfrak{T}^H(\boldsymbol{\nu},\boldsymbol{c})
\mathfrak{T}(\boldsymbol{\nu},\boldsymbol{c})
\boldsymbol{\Sigma}_\text{X}}{\partial \nu_u}\right\}\notag\\
&\quad = \operatorname{Tr}\left\{
	\begin{bmatrix}
		\boldsymbol{0},
		\mathfrak{T}^H(\boldsymbol{\nu},\boldsymbol{c}) \boldsymbol{\Xi}_u\boldsymbol{c},\boldsymbol{0}
	\end{bmatrix}\boldsymbol{\Sigma}_\text{X}\right\} +\operatorname{Tr}\left\{\begin{bmatrix}
		\boldsymbol{0},
		\mathfrak{T}^H(\boldsymbol{\nu},\boldsymbol{c})
		 \boldsymbol{\Xi}_u\boldsymbol{c},\boldsymbol{0}
	\end{bmatrix}^\text{H} \boldsymbol{\Sigma}_\text{X}
\right\}\notag\\
&\quad=2\mathcal{R}\left\{\sum_{m=0}^{U-1}\boldsymbol{c}^\text{H}\boldsymbol{\Psi}^\text{H}_{m}(\boldsymbol{\nu})\boldsymbol{\Xi}_u\boldsymbol{c}\Sigma_{\text{X},u,m}\right\} =2\mathcal{R}\left\{
\left[\mathfrak{T}^H(\boldsymbol{\nu},\boldsymbol{c}) \boldsymbol{\Xi}_u\boldsymbol{c}\right]^H\boldsymbol{\Sigma}_{\text{X},:,u}
\right\}.
\end{align} 
$\frac{\partial \mathcal{G}_{1} (\boldsymbol{c},\boldsymbol{\nu})}{\partial \boldsymbol{\nu}} $ can be simplified as
\begin{align}
\frac{\partial \mathcal{G}_{1} (\boldsymbol{c},\boldsymbol{\nu})}{\partial \boldsymbol{\nu}}  = 2\mathcal{R}\left\{\operatorname{diag}\left\{
\boldsymbol{\Sigma}_\text{X}
\mathfrak{T}^H(\boldsymbol{\nu},\boldsymbol{c})
\boldsymbol{\Xi}(\boldsymbol{I}_U\otimes \boldsymbol{c})\right\}^\text{T}
	\right\}.
\end{align}  
Additionally, $\frac{\partial \mathcal{G}_{2,m}(\boldsymbol{c},\boldsymbol{\nu})}{\partial \boldsymbol{\nu}}$ can be obtained as
\begin{align}\label{gf}
&\frac{\partial \mathcal{G}_{2,m}(\boldsymbol{c},\boldsymbol{\nu})}{\partial \boldsymbol{\nu}}  = -2\mathcal{R}\left\{[\boldsymbol{y}_m-\boldsymbol{\Psi}(\boldsymbol{\nu})(\boldsymbol{\mu}_m\otimes \boldsymbol{c})]^\text{H}\frac{\partial \boldsymbol{\Psi}(\boldsymbol{\nu})(\boldsymbol{\mu}_m\otimes \boldsymbol{c})}{\partial \boldsymbol{\nu}}\right\}\\
& = -2\mathcal{R}\left\{[\boldsymbol{y}_m-\boldsymbol{\Psi}(\boldsymbol{\nu})(\boldsymbol{\mu}_m\otimes \boldsymbol{c})]^\text{H}  \boldsymbol{\Xi}(\operatorname{diag}\{\boldsymbol{\mu}_m\}\otimes \boldsymbol{c})
\right\}.\notag
	\end{align}

\begin{algorithm}
\setstretch{1.7}
	\caption{DFSMC algorithm for direction finding with the unknown mutual coupling effect} \label{alg1}
	\begin{algorithmic}[1]
		\STATE  \emph{Input:} received signal $\boldsymbol{Y}$, the number of samples $M$, the numbers  of iterations $N_1$, $N_2$ and $N_3$ dictionary matrix $\boldsymbol{D}$, the first order derivative of dictionary matrix $\boldsymbol{\Xi}$. Usually, we have $N_1=10^3$, $N_2=300$, $N_3=50$, $b=d=f=10^{-3}$, and $a=c=e=1+b$.
		\STATE \emph{Initialization:} $i_{\text{iter}}=1$, $s_{\text{method}}=0$, $i_{\text{method}}=1$, $\hat{\boldsymbol{c}}=\hat{\boldsymbol{\vartheta}}=[1,\boldsymbol{0}_{1\times (N-1)}]^T$, and $\hat{\boldsymbol{\nu}}=\boldsymbol{0}_{U\times 1}$.
		\STATE $\boldsymbol{\Psi}(\hat{\boldsymbol{\nu}})\leftarrow \boldsymbol{D}+\boldsymbol{\Xi}\left(\operatorname{diag}\left\{
	\hat{\boldsymbol{\nu}} \right\}\otimes \boldsymbol{I}_{N}\right)$.
		\WHILE{$i_{\text{iter}} \leq N_{1}$ }
		\STATE Obtain $\mathfrak{T}(\boldsymbol{\nu},\boldsymbol{c})$ from (\ref{TT}).
		\STATE \label{5} Obtain the mean $\boldsymbol{\mu}_m$ ($m=0,1,\dots,M-1$) and covariance matrix $\boldsymbol{\Sigma}_{\text{X}}$ from (\ref{x}) and (\ref{sigma}), respectively.
		\STATE Update the precision of noise variance $\hat{\alpha}_{\text{n}}$ from (\ref{n}).
		\STATE Update the precision of signal variance $\hat{\boldsymbol{\iota}}$ from (\ref{beta}).
		\STATE Obtain the spatial spectrum $
			P_\text{X}=\begin{bmatrix}\frac{1}{\hat{\iota}_0},\frac{1}{\hat{\iota}_1},\dots,\frac{1}{\hat{\iota}_{N-1}}\end{bmatrix}^\text{T}$.
		\IF{$i_{\text{iter}}\geq N_2$ and $s_{\text{method}}=1$}
		\STATE $i_{\text{method}}\leftarrow i_{\text{method}}+1$.
		\IF{$i_{\text{method}}=N_3$}
		\STATE $i_{\text{method}}\leftarrow 1$.
		\STATE $s_{\text{method}}\leftarrow 0$.
		\ENDIF
		\STATE \label{10} Update the off-grid vector $\hat{\boldsymbol{\nu}}$ from (\ref{nu}).
		\STATE $\boldsymbol{\Psi}(\hat{\boldsymbol{\nu}})\leftarrow \boldsymbol{D}+\boldsymbol{\Xi}\left(\operatorname{diag}\left\{
	\hat{\boldsymbol{\nu}} \right\}\otimes \boldsymbol{I}_{N}\right)$.
		\ENDIF	
		
		\IF{$i_{\text{iter}}\geq N_2$ and $s_{\text{method}}=0$}
		\STATE $i_{\text{method}}\leftarrow i_{\text{method}}+1$.
		\IF{$i_{\text{method}}=N_3$}
		\STATE $i_{\text{method}}\leftarrow 1$.
		\STATE $s_{\text{method}}\leftarrow 1$.
		\ENDIF
		\STATE Update the precision of mutual coupling variance $\hat{\boldsymbol{\vartheta}}$ from (\ref{varT}).
		\STATE \label{8} Update the mutual coupling vector $\hat{\boldsymbol{c}}$ from (\ref{ctE}).
		\ENDIF
		\STATE $i_{\text{iter}}\leftarrow i_{\text{iter}}+1$.
		\ENDWHILE
		\STATE \emph{Output:} the spatial spectrum $P_{\text{X}}$, and the directions $(\boldsymbol{\zeta}+\boldsymbol{\nu})$.
	\end{algorithmic}
\end{algorithm}

Therefore, with $\frac{\partial \mathcal{L}(\boldsymbol{\nu})}{\partial \nu_u}=0$, we can obtain  
\begin{align}\label{nu}
	\hat{\boldsymbol{\nu}}=\boldsymbol{G}^{-1}\boldsymbol{z},
\end{align}
where the entry of the $u$-th row in $\boldsymbol{G}\in\mathbb{R}^{U\times U}$ is
\begin{align}
	\boldsymbol{G}_{u,:}=& \mathcal{R}\Big\{M\boldsymbol{c}^H
\boldsymbol{\Xi}_u^H\boldsymbol{\Xi}
(\operatorname{diag}\left\{
\boldsymbol{\Sigma}_{\text{X},:,u}
	 \right\}\otimes \boldsymbol{c})\Big\} 
+\sum^{M-1}_{m=0} \mathcal{R}\Big\{\mu_{m,u} 
 \boldsymbol{c}^\text{T}
 \boldsymbol{\Xi}^T_u\boldsymbol{\Xi}^\text{*}(\boldsymbol{I}_U\otimes\boldsymbol{c}^*)
 \operatorname{diag}^*\left\{
	 \boldsymbol{\mu}_m \right\} 
	 \Big\},
\end{align}
and the $u$-th entry of $\boldsymbol{z}\in\mathcal{R}^{U\times 1}$ is
\begin{align}
	z_u 
& = \sum^{M-1}_{m=0} \mathcal{R}\Big\{[\boldsymbol{y}_m-\boldsymbol{D}(\boldsymbol{\mu}_m\otimes \boldsymbol{c})]^\text{H} \boldsymbol{\Xi}_u\mu_{m,u}\boldsymbol{c}\Big\}  -M\mathcal{R}\left\{
\boldsymbol{c}^H
\boldsymbol{\Xi}_u^H\boldsymbol{D}(\boldsymbol{I}_{U}
	\otimes \boldsymbol{c})
\boldsymbol{\Sigma}_{\text{X},:,u}\right\}.
\end{align}


In Algorithm~\ref{alg1}, we show the details of the proposed DFSMC method for the direction finding with the unknown mutual coupling effect. In the proposed DFSMC algorithm, after the iterations, we can obtain the spatial spectrum $P_{\text{X}}$ of the sparse matrix $\boldsymbol{X}$ from the received signal $\boldsymbol{Y}$. Then, by searching all the values of  $P_{\text{X}}$, the corresponding peak values can be found. By selecting positions of peak values corresponding to the $K$ maximum values, we can estimate the directions with $\boldsymbol{\zeta}+\boldsymbol{\nu}$.

\section{Simulation Results}\label{sec4}
\begin{table}[!t]
	\renewcommand{\arraystretch}{1.3}
	\caption{Simulation Parameters}
	\label{table1}
	\centering
	\begin{tabular}{cc}
		\hline
		\textbf{Parameter} & \textbf{Value}\\
		\hline
		The signal-to-noise ratio (SNR) & $ 20 $ dB\\
		The number of samples $M$ & $100$\\
		The number of antennas $N$ & $20$\\
		The number of signals $K$& $3$\\
		The space between antennas $d$& $0.5$ wavelength\\
		The grid space $\delta$ & $\ang{1}$\\
		The direction range & $\left[\ang{-60},\ang{60}\right]$\\
		The hyperparameters $b,d,f$ & $10^{-3}$\\
		 {\hll $N_1$ in Algorithm}~\ref{alg1} & $N_1=10^3$ \\
		{\hll $N_2$  in Algorithm~\ref{alg1}} & $N_2=300$ \\
		{\hll $N_3$ in Algorithm~\ref{alg1}} & $N_3=50$\\
	\hline
	\end{tabular}
\end{table}

\begin{figure}
	\centering
	\includegraphics[width=3.5in]{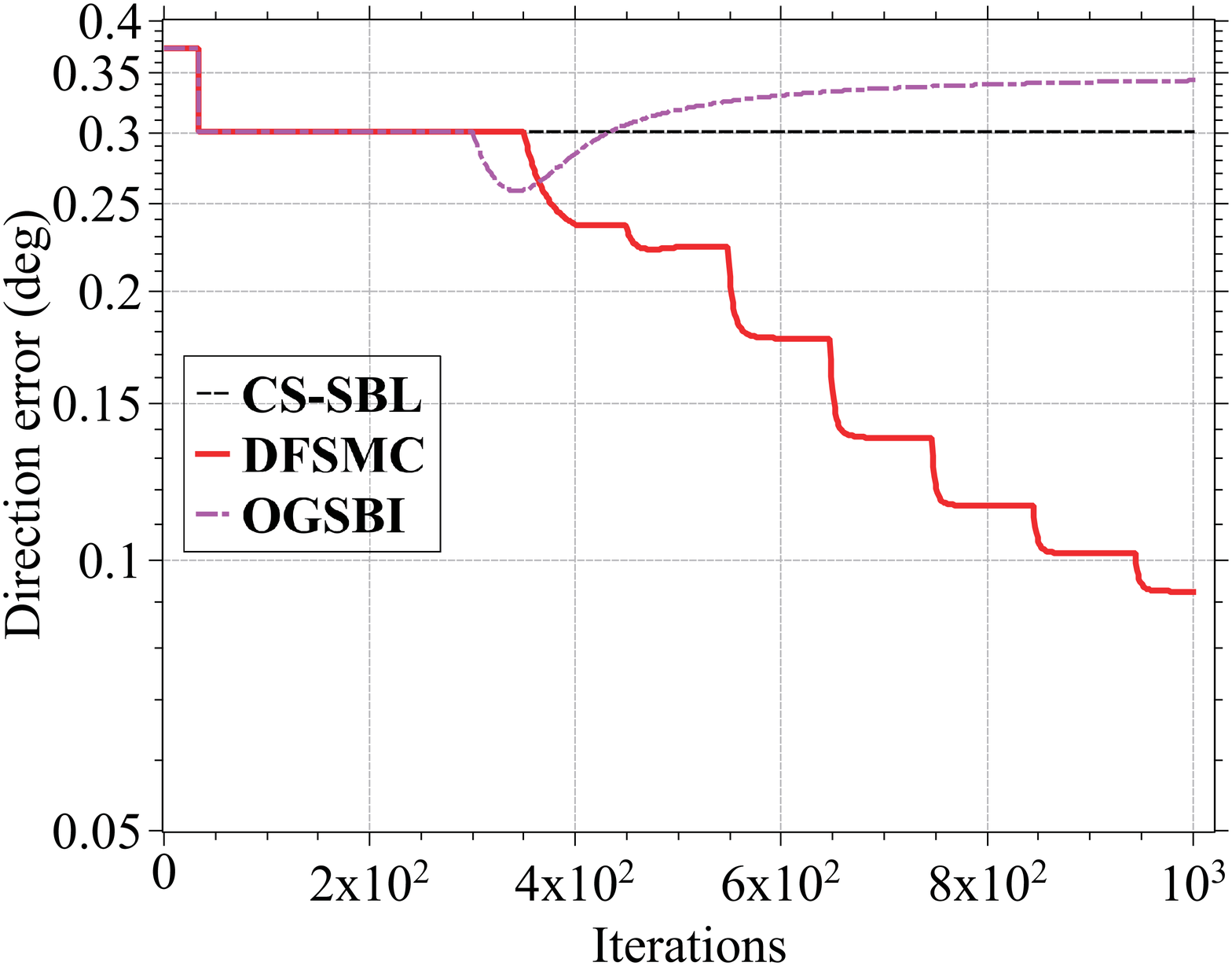}
	\caption{The estimation error with iterations ($\alpha_{\text{c}}=-8$ dB).}
	\label{iter8}
\end{figure}

\begin{figure}
	\centering
	\includegraphics[width=3.5in]{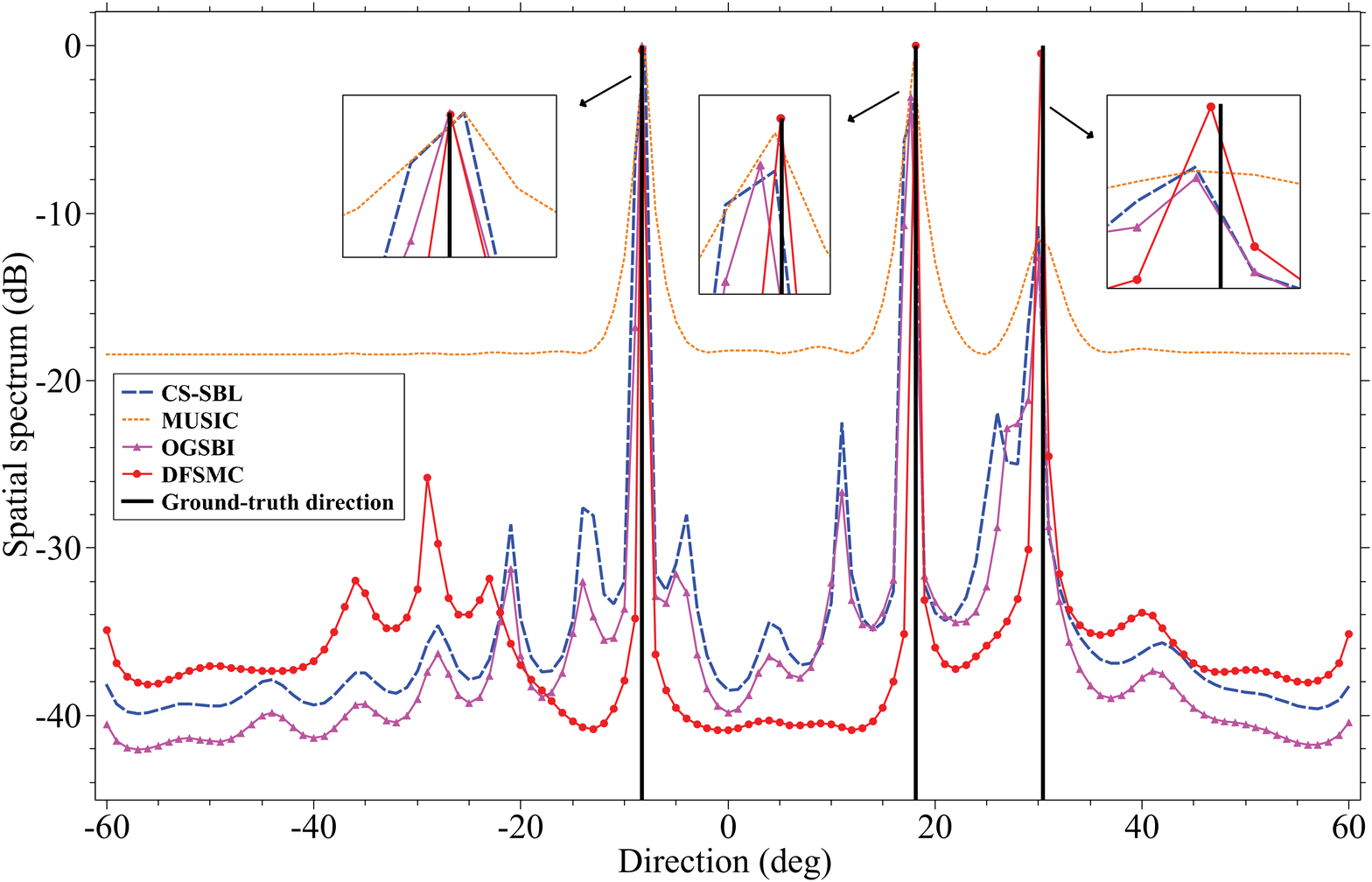}
	\caption{The spatial spectrum for direction estimation ($\alpha_{\text{c}}=-8$ dB).}
	\label{sp8}
\end{figure}

\begin{figure}
	\centering
	\includegraphics[width=3.5in]{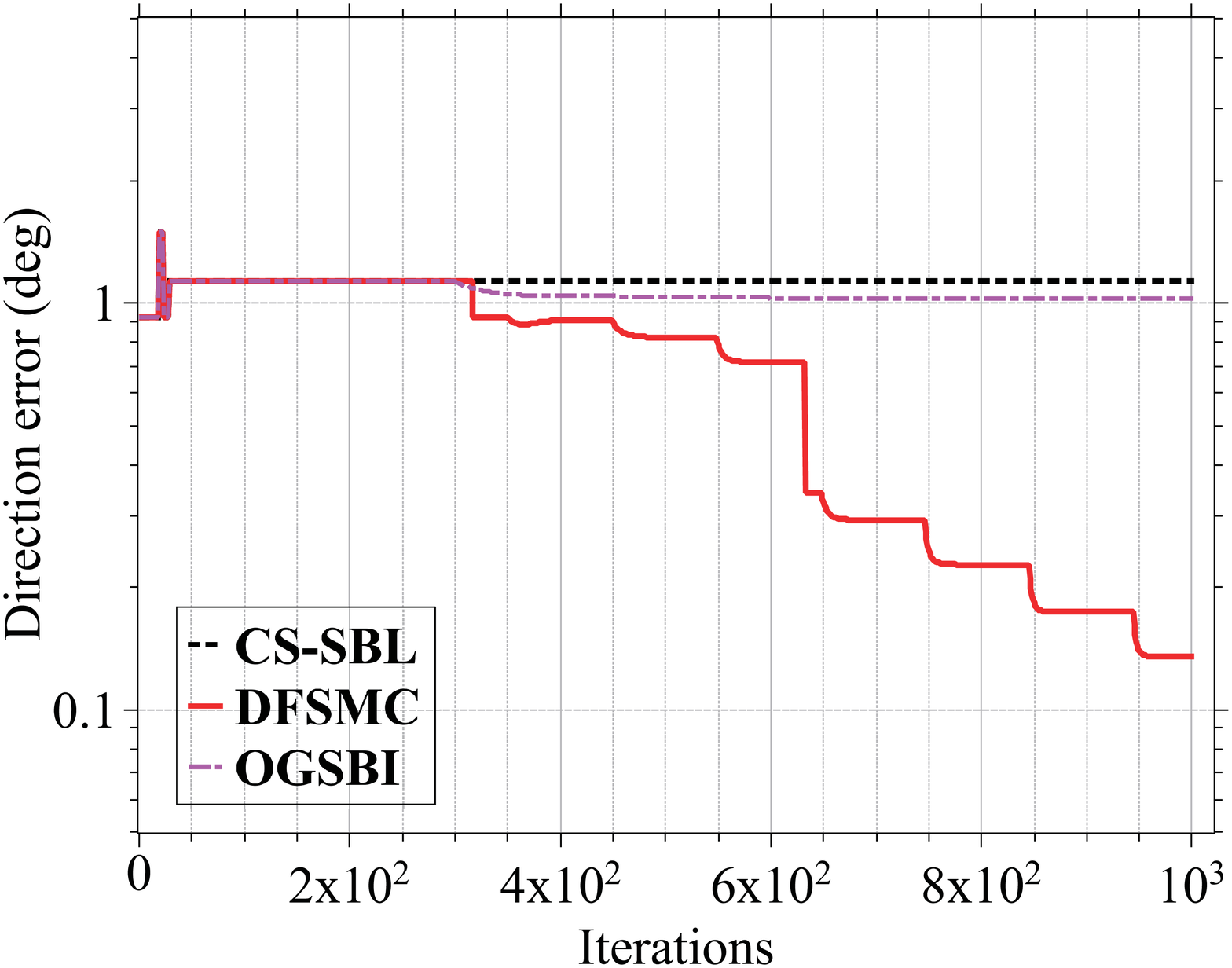}
	\caption{The estimation error with iterations ($\alpha_{\text{c}}=-5$ dB).}
	\label{iter}
\end{figure}

\begin{figure}
	\centering
	\includegraphics[width=3.5in]{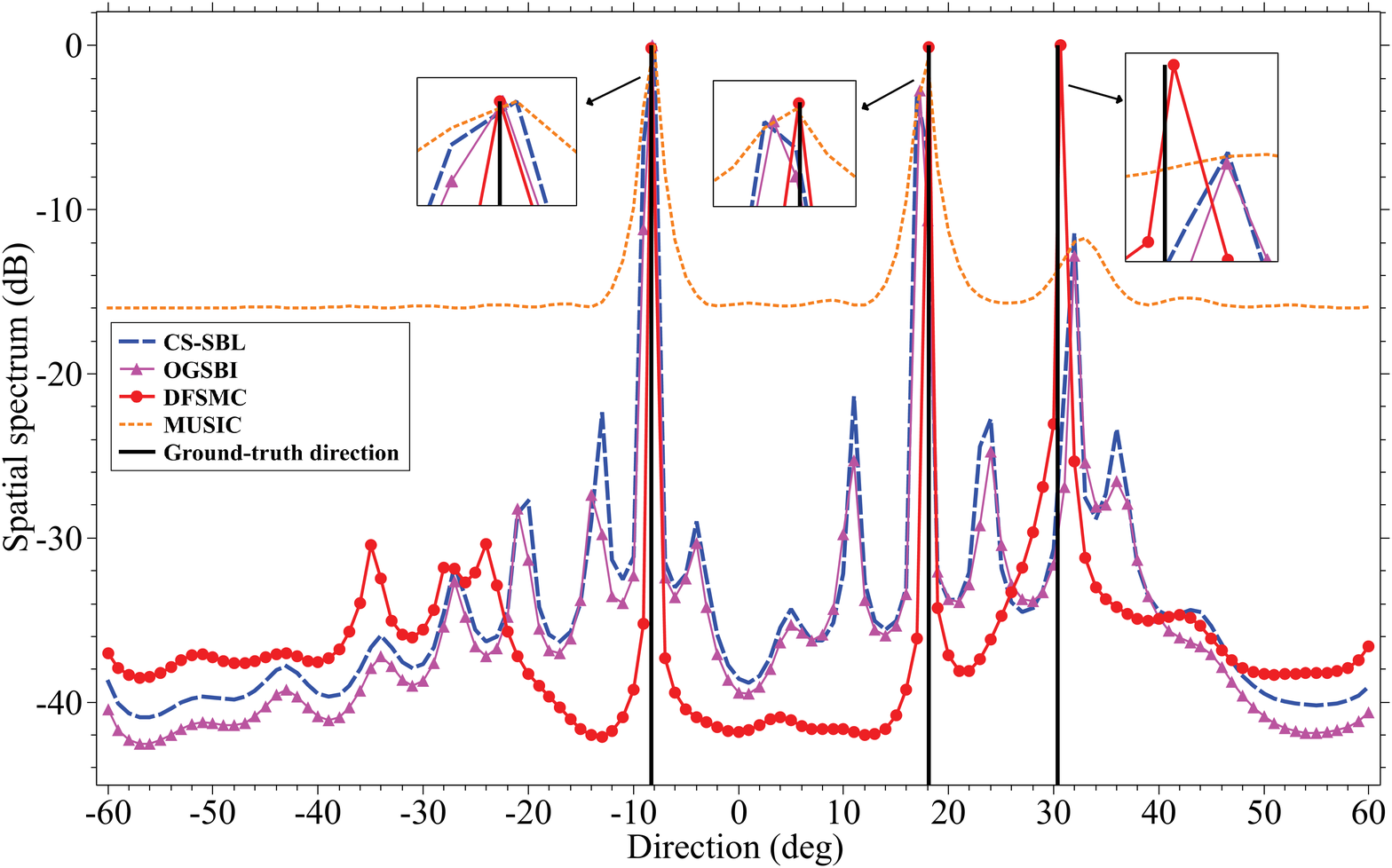}
	\caption{The spatial spectrum for direction estimation ($\alpha_{\text{c}}=-5$ dB).}
	\label{sp}
\end{figure} 


\begin{figure}
	\centering
	\includegraphics[width=3.5in]{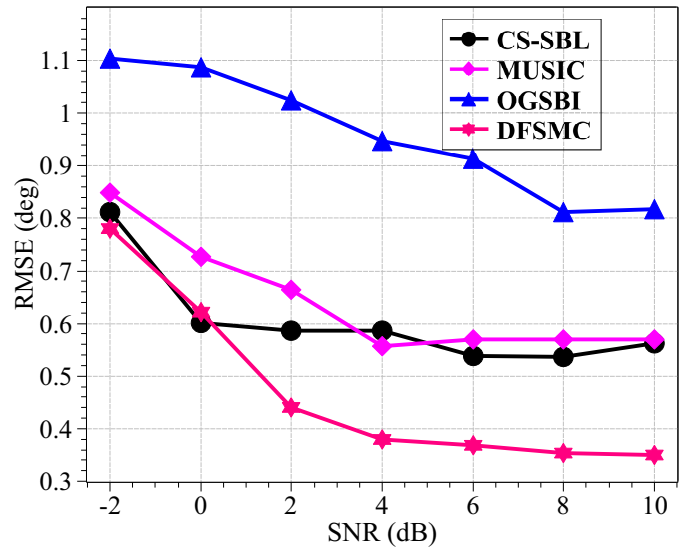}
	\caption{The direction estimation performance with different SNRs.}
	\label{SNR}
\end{figure}

\begin{figure}
	\centering
	\includegraphics[width=3.5in]{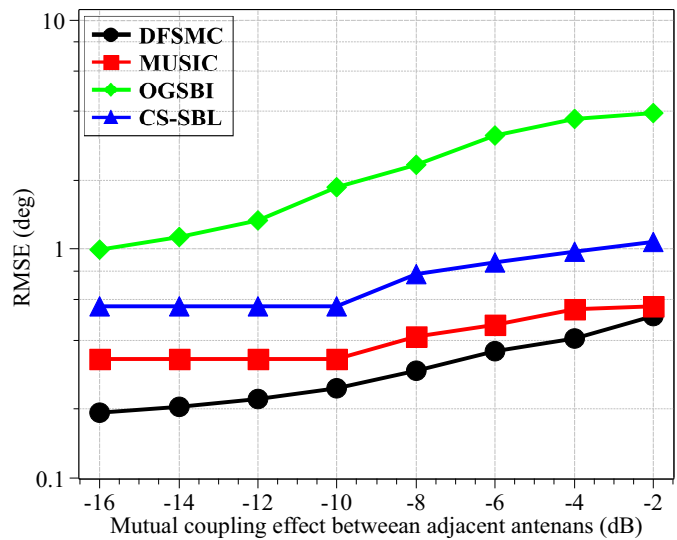}
	\caption{The direction estimation performance with different mutual coupling effect.}
	\label{MC}
\end{figure}


\begin{table}[!t]
	\renewcommand{\arraystretch}{1.3}
	\caption{Estimated directions ($\alpha_{\text{c}}=-8$ dB)}
	\label{table2}
	\centering
	\begin{tabular}{cccc}
		\hline
		\textbf{Methods} & \textbf{Signal $1$} & \textbf{Signal $2$} & \textbf{Signal $3$}\\
		\hline
		\textbf{Ground-truth directions} &$\ang{-8.268}$& $\ang{18.128}$& $\ang{30.428}$\\
		\textbf{OGSBI}& $\ang{-8.267}$ & $\ang{17.69}$&
  $\ang{30.02}$\\
		\textbf{CS-SBL}& $\ang{-8}$ & $\ang{18}$&
    $\ang{30}$\\
		\textbf{MUSIC}&
		$\ang{-8}$&$\ang{18}$&$\ang{30}$\\
		\textbf{DFSMC}&$\ang{-8.254}$&$\ang{18.13}$&$\ang{30.27}$\\
	\hline
	\end{tabular}
\end{table}

\begin{table}[!t]
	\renewcommand{\arraystretch}{1.3}
	\caption{Estimated directions ($\alpha_{\text{c}}=-5$ dB)}
	\label{table3}
	\centering
	\begin{tabular}{cccc}
		\hline
		\textbf{Methods} & \textbf{Signal $1$} & \textbf{Signal $2$} & \textbf{Signal $3$}\\
		\hline
		\textbf{Ground-truth directions} &$\ang{-8.268}$& $\ang{18.128}$& $\ang{30.428}$\\
		\textbf{OGSBI}& $\ang{-8.222}$ & $\ang{17.29}$&
  $\ang{31.99}$\\
		\textbf{CS-SBL}& $\ang{-8}$ & $\ang{17}$&
    $\ang{32}$\\
		\textbf{MUSIC}&
		$\ang{-8}$&$\ang{18}$&$\ang{33}$\\
		\textbf{DFSMC}&$\ang{-8.260}$&$\ang{18.11}$&$\ang{30.66}$\\
	\hline
	\end{tabular}
\end{table}

Extensive simulation results have been conducted. All experiments are conducted in Matlab R2017b on a PC with a 2.9 GHz Intel Core i5 and 8 GB of RAM, and Matlab codes have been made available online at \url{https://drive.google.com/drive/folders/1XwzbNtHXfjTrN4-wylAhGI-CwIY3u3K1}. The mutual coupling effect can be generated by the following expression
\begin{align}
	c_n=\begin{cases}
		(1+\xi_{\text{c}})e^{j\phi_{\text{c}}}10^{\frac{\alpha_{\text{c}}(1+0.5n)}{20}},&n<5 \\
		0,&\text{otherwise} 
	\end{cases},
\end{align}
where $\xi_{\text{c}}\sim \mathcal{U}_{\xi_{\text{c}}}\left(\left[-0.05,0.05\right]\right)$, $\phi_{\text{c}}\sim \mathcal{U}_{\phi_{\text{c}}}\left(\left[0,2\pi\right]\right)$, and we use the parameter $\alpha_{\text{c}}$ in dB to measure the mutual coupling effect between adjacent antennas. Additionally, we use the independent Gaussian distribution to generate the received signals, and for the $m$-th sample in the $n$-th antenna, we have $s_{n,m}\sim\mathcal{CN}\left(\sqrt{2}e^{j\frac{\pi}{2}},1\right)$.  

In this paper, to compare with the state-of-art direction estimation methods, we compare the proposed DFSMC method with the following algorithms:
\begin{itemize}
	\item CS-SBL\footnote{The MATLAB code was downloaded at \url{http://people.ee.duke.edu/~lcarin/BCS.html}}, the Bayesian compressive sensing method proposed in~\cite{shihao2008}.
	\item OGSBI\footnote{The MATLAB code was downloaded at \url{https://sites.google.com/site/zaiyang0248/publication}}, the off-grid sparse Bayesian inference method proposed in~\cite{yang2013}.
	\item MUSIC, the multiple signal classification method~\cite{ralph1986,schmidt1981}.
\end{itemize}

With the simulation parameters in Table~\ref{table1} and the 
        the mutual coupling between adjacent antennas being $\alpha_{\text{c}}=-8$ dB, the spatial spectrum is given in Fig.~\ref{sp8}, where the proposed DFSMC is compared with MUSIC, CS-SBL and OGSBI methods. The estimated directions for $K=3$ signals are given in Table~\ref{table2}. Additionally, the iteration processes of DFSMC, CS-SBL, and OGSBI methods are also given in Fig.~\ref{iter8}. With both the mutual coupling effect and off-grid, the proposed DFSMC method is advantageous in this scenario. 

As shown in Fig.~\ref{iter8}, in the first $300$ iterations, DFSMC method only updates the parameters $\boldsymbol{\mu}_m$, $\boldsymbol{\Sigma}_{\text{X}}$ and $\boldsymbol{\iota}$. Then, during the $301$ to $350$ iterations, the mutual coupling parameters $\boldsymbol{c}$ and $\boldsymbol{\vartheta}$ are updated. For the next $50$ iterations, the off-grid parameter $\boldsymbol{\nu}$ is updated. With repeating the $50$ iterations to update the mutual coupling parameters and the off-grid parameter, the direction error can be decreased. Moreover, as shown in Fig.~\ref{iter8}, when only the mutual coupling parameters are updated, the direction estimation performance can be not improved with the correct estimated directions. However, for the next off-grid estimation, the better performance can be achieved with the updated mutual coupling parameters.   

The estimated spatial spectrum is shown in Fig.~\ref{sp8}. It can be seen that the positions of peak spectrum are closer to the ground-truth directions using the DFSMC method than the OGSBI, CS-SBL and MUSIC methods. The corresponding estimated directions are given in Table~\ref{table2}. When we use the following expression to measure the estimation performance
\begin{align}
    e_1=\sqrt{\frac{1}{K}\|\hat{\boldsymbol{\theta}}-\boldsymbol{\theta}\|^2_2},
\end{align}
where $\hat{\boldsymbol{\theta}}$ denotes the estimated directions. Then, the estimation errors (in deg) of DFSMC, OGSBI, CS-SBL and MUSIC methods can be obtained as $\ang{0.092}$, $\ang{0.346}$, $\ang{0.301}$ and $\ang{0.301}$, respectively. Therefore, since the mutual coupling effect is estimated in the proposed DFSMC method, the direction estimation performance is much better than the existing methods including OGSBI, CS-SBL and MUSIC.

When the mutual coupling effect increases from $\alpha_{\text{c}}=-8$ dB to $\alpha_{\text{c}}=-5$ dB, the corresponding iteration processes and the spatial spectrums of DFSMC, CS-SBL, OGSBI and MUSIC methods are given in Fig.~\ref{iter} and Fig.~\ref{sp}. The estimated directions are given inTable~\ref{table3}, and the estimation errors (in deg) of DFSMC, OGSBI, CS-SBL and MUSIC methods can be obtained as $\ang{0.134}$, $\ang{1.024}$, $\ang{1.128}$ and $\ang{1.495}$. Compared with the direction estimation  performance in the scenario $\alpha_{\text{c}}=-8$ dB, the performance in the scenario $\alpha_{\text{c}}=-5$ dB decreases for all the methods, so the mutual coupling effect has a great effect on the direction estimation performance. However, the proposed DFSMC method can also achieve much better performance than existing methods.

With the $100$ trails, the direction estimation performance with different SNRs is given in Fig.~\ref{SNR}, where we use the following root-mean-square error (RMSE) expression to measure the estimation performance
\begin{align}
    e_2=\sqrt{\frac{1}{KP}\sum_{p=0}^{P-1}\|\hat{\boldsymbol{\theta}}_p-\boldsymbol{\theta}_p\|^2_2},
\end{align}
where $P$ denote the number of trails, $\boldsymbol{\theta}_p$ denotes the directions in the $p$-th trail, and $\hat{\boldsymbol{\theta}}_p$ denotes the estimated directions in the $p$-th trail. As shown in Fig.~\ref{SNR}, the proposed DFSMC method achieves the best estimation performance when the SNR of received signals is greater than $0$ dB. Almost the same estimation performance is achieved by the MUSIC and CS-SBL method. However, with the mutual coupling effect, the direction grids usually cannot be estimated correctly, so the further off-grid optimization in OGSBI cannot improve the estimation performance. Fig.~\ref{SNR} indicates that our proposed DFSMC method is very advantageous in the cases when the SNR of received signals is large.

With different mutual coupling effects between antennas, we show the simulation results in Fig.~\ref{MC}, where the mutual coupling effect $\alpha_{\text{c}}$ between adjacent antennas is from $-16$ dB to $-2$ dB. Since the proposed DFSMC method estimates the mutual coupling vector $\boldsymbol{c}$ iteratively, and DFSMC achieves the best estimation performance among the existing methods including CS-SBL, OGSBI, and MUSIC. It can be seen that with optimizing the off-grid and the mutual coupling vector, the performance of direction estimation can be improved by estimating the sparse signals in the continue domain using the DFSMC method. {\hll The computational complexity of the proposed algorithm mainly depends on step~\ref{5}, step~\ref{10} and step~\ref{8}. The computational complexity of  step~\ref{5}, step~\ref{10} and step~\ref{8} can be obtained as $\mathcal{O}(MU^2N+UN^2+U^3)$, $\mathcal{O}(U^3+MNU^2)$ and $\mathcal{O}(U^3+MU^2N)$. Therefore, the computational complexity of the proposed algorithm can be obtained as $\mathcal{O}(U^3+MU^2N+UN^2)$. Additionally, with $U\geq N$, the computational complexity can be simplified as $\mathcal{O}(MUN^2)$. The computational complexity of the proposed algorithm has the same order of the SBL-based algorithms, such as the OGSBI algorithm and the SBL algorithm. 
}

\section{Conclusions}\label{sec5}
The direction finding problem with the unknown mutual coupling effect has been investigated in this paper. The novel DFSMC method has been proposed to estimate the directions, the means, and variance of received signals, the mutual coupling vector, the noise variance, and the off-grid vector, et al. iteratively. Additionally, the expressions to estimate the unknown parameters have been theoretically derived using the EM method.  Simulation results confirm that the proposed DFSMC method outperforms the existing direction finding methods in the ULA system with the unknown mutual coupling effect.  Future work will focus on the extension of the proposed DFSMC method in the scenario with correlated signals.

\vspace{6pt}

\authorcontributions{conceptualization, Peng Chen and Zhimin Chen; methodology, Peng Chen and Xuan Zhang; software, Zhimin Chen; validation, Linxi Liu; formal analysis, Zhimin Chen; investigation, Peng Chen; resources, Peng Chen; data curation, Liuxi Liu; writing—original draft preparation, Peng Chen; writing—review and editing, Zhimin Chen; visualization, Zhimin Chen; supervision, Xuan Zhang; project administration, Peng Chen; funding acquisition, Peng Chen.}
 
\acknowledgments{This work was supported in part by the  National Natural Science Foundation of China (Grant No. 61801112, 61601281), the Natural Science Foundation of Jiangsu Province (Grant No. BK20180357), the Open Program of State Key Laboratory of Millimeter Waves at Southeast University (Grant No. Z201804).}

\conflictsofinterest{The authors declare no conflict of interest.} 
 
\appendixtitles{yes} 
\appendixsections{one} 
\appendix
\section{Proof of Lemma~\ref{lemma}}\label{derivation}
When the complex vectors $\boldsymbol{u}$ and $\boldsymbol{v}$ are the functions of $\boldsymbol{x}$, we can obtain
\begin{align}
\frac{\partial\boldsymbol{u}^\text{H}\boldsymbol{v}  }{\partial \boldsymbol{x}} & = \begin{bmatrix}
\frac{\partial\boldsymbol{u}^\text{H}\boldsymbol{v}  }{\partial x_0} ,\frac{\partial\boldsymbol{u}^\text{H}\boldsymbol{v}  }{\partial x_1},\ldots, \frac{\partial\boldsymbol{u}^\text{H}\boldsymbol{v}  }{\partial x_{N-1}}
\end{bmatrix}  = \begin{bmatrix}\frac{\partial\sum^{M-1}_{m=0}u^*_mv_m  }{\partial x_0},
\ldots, \frac{\partial\sum^{M-1}_{m=0}u^*_mv_m  }{\partial x_n},\ldots
\end{bmatrix}\notag\\
& = \begin{bmatrix}
\ldots, \sum^{M-1}_{m=0} \frac{\partial u^*_m  }{\partial x_n} v_m+u^*_m\frac{\partial v_m  }{\partial x_n},\ldots
\end{bmatrix}   = \begin{bmatrix}
\ldots, \left(\frac{\partial \boldsymbol{u}^*}{\partial x_n}\right)^\text{T}\boldsymbol{v}+
\boldsymbol{u}^\text{H}\frac{\partial \boldsymbol{v}}{\partial x_n},\ldots
\end{bmatrix}  \notag\\
& = \boldsymbol{v}^\text{T}
\begin{bmatrix}\frac{\partial \boldsymbol{u}^*}{\partial x_0},
\ldots, \frac{\partial \boldsymbol{u}^*}{\partial x_n}, \ldots
\end{bmatrix}
+\boldsymbol{u}^\text{H} \begin{bmatrix}\frac{\partial \boldsymbol{v}}{\partial x_0},
\ldots, 
\frac{\partial \boldsymbol{v}}{\partial x_n},\ldots
\end{bmatrix} \notag\\
& =\boldsymbol{v}^\text{T} \frac{\partial (\boldsymbol{u}^*)}{\partial \boldsymbol{x}} 
+\boldsymbol{u}^\text{H} 
\frac{\partial \boldsymbol{v}}{\partial \boldsymbol{x}}.
\end{align}

With $\boldsymbol{A}$ and $\boldsymbol{u}$ being the function of $\boldsymbol{x}$, we can obtain the  entry in $m$-th row and $n$-th column of $\frac{\partial \boldsymbol{Au}}{\partial \boldsymbol{x}}$ as
\begin{align}
\frac{\partial \left[\boldsymbol{Au}\right]_m}{\partial x_n}& =\frac{\partial
	\sum_{p=0}^{P-1}A_{m,p}u_p }{\partial x_n} =\sum_{p=0}^{P-1}\frac{\partial A_{m,p}}{\partial x_n}u_p+A_{m,p}\frac{\partial u_p}{\partial x_n} \\
& =\boldsymbol{u}^\text{T}
\frac{\partial [\boldsymbol{A}^\text{T}]_m}{\partial x_n}+[\boldsymbol{A}^\text{T}]^\text{T}_m\frac{\partial \boldsymbol{u}}{\partial x_n} = \left[
\frac{\partial \boldsymbol{A}}{\partial x_n}\boldsymbol{u} + \boldsymbol{A}\frac{\partial \boldsymbol{u}}{\partial x_n}\right]_m,\notag
\end{align}
so the $n$-th column of $\frac{\partial \boldsymbol{Au}}{\partial \boldsymbol{x}}$  is
\begin{align}\left[\frac{\partial \boldsymbol{Au}}{\partial \boldsymbol{x}}\right]_n=
\frac{\partial \boldsymbol{A}}{\partial x_n}\boldsymbol{u} + \boldsymbol{A}\frac{\partial \boldsymbol{u}}{\partial x_n}, \end{align} 
and
\begin{align}
\frac{\partial \boldsymbol{Au}}{\partial \boldsymbol{x}}= \begin{bmatrix}
\frac{\partial \boldsymbol{A}}{\partial x_0}\boldsymbol{u} + \boldsymbol{A}\frac{\partial \boldsymbol{u}}{\partial x_0},
\ldots, \frac{\partial \boldsymbol{A}}{\partial x_n}\boldsymbol{u} + \boldsymbol{A}\frac{\partial \boldsymbol{u}}{\partial x_n} ,\ldots
\end{bmatrix}.
\end{align}

\reftitle{References}
 
\externalbibliography{yes}
\bibliography{IEEEabrv,References}   

\begin{thebibliography}{-------}
\providecommand{\natexlab}[1]{#1}

\bibitem[Schmidt(1986)]{ralph1986}
Schmidt, R.O.
\newblock {Multiple emitter location and signal parameter estimation}.
\newblock {\em {IEEE} Trans. Antennas Propag.} {\bf 1986}, {\em 34},~276--280.

\bibitem[Schmidt(1981)]{schmidt1981}
Schmidt, R.
\newblock {A signal subspace approach to multiple emitter location spectrum
  estimation}.
\newblock PhD thesis, Stanford University, Stanford, CA,  1981.

\bibitem[Zoltowski \em{et~al.}(1993)Zoltowski, Kautz, and
  Silverstein]{Zoltowski1993}
Zoltowski, M.; Kautz, G.; Silverstein, S.
\newblock {Beamspace Root-MUSIC}.
\newblock {\em {IEEE} Trans. Signal Process.} {\bf 1993}, {\em 41},~344--364.

\bibitem[Roy and Kailath(1989)]{roy1989}
Roy, R.; Kailath, T.
\newblock {ESPRIT-estimation of signal parameters via rotational invariance
  techniques}.
\newblock {\em {IEEE} Trans. Acoust., Speech, Signal Process.} {\bf 1989}, {\em
  37},~984--995.

\bibitem[Pham \em{et~al.}(2016)Pham, Loubaton, and Vallet]{Pham2017}
Pham, G.T.; Loubaton, P.; Vallet, P.
\newblock {Performance analysis of spatial smoothing schemes in the context of
  large arrays}.
\newblock {\em {IEEE} Trans. Signal Process.} {\bf 2016}, {\em 64},~160 --172.

\bibitem[Chen \em{et~al.}(2018)Chen, Cao, Chen, and Yu]{pengele}
Chen, P.; Cao, Z.; Chen, Z.; Yu, C.
\newblock {Sparse DOD/DOA estimation in a bistatic MIMO radar with mutual
  coupling effect}.
\newblock {\em {Electronics}} {\bf 2018}, {\em 7},~341.

\bibitem[Carlin \em{et~al.}(2013)Carlin, Rocca, Oliveri, Viani, and
  Massa]{Matteo2013}
Carlin, M.; Rocca, P.; Oliveri, G.; Viani, F.; Massa, A.
\newblock {Directions-of-arrival estimation through Bayesian compressive
  sensing strategies}.
\newblock {\em {IEEE} Trans. Antennas Propag.} {\bf 2013}, {\em 61},~3828 --
  3838.

\bibitem[Chen \em{et~al.}(2017)Chen, Qi, Wu, and Wang]{7467561}
Chen, P.; Qi, C.; Wu, L.; Wang, X.
\newblock {Estimation of Extended Targets Based on Compressed Sensing in
  Cognitive Radar System}.
\newblock {\em {IEEE Transactions on Vehicular Technology}} {\bf 2017}, {\em
  66},~941--951.

\bibitem[Yu \em{et~al.}(2011)Yu, Petropulu, and Poor]{yao2011}
Yu, Y.; Petropulu, A.P.; Poor, H.V.
\newblock {Measurement matrix design for compressive sensing-based MIMO radar}.
\newblock {\em {IEEE} Trans. Signal Process.} {\bf 2011}, {\em 59},~5338 --
  5352.

\bibitem[Carlin \em{et~al.}(2016)Carlin, Rocca, Oliveri, Viani, and
  Massa]{Matteo2016}
Carlin, M.; Rocca, P.; Oliveri, G.; Viani, F.; Massa, A.
\newblock {Novel wideband DOA estimation based on sparse Bayesian learning with
  dirichlet process priors}.
\newblock {\em {IEEE} Trans. Signal Process.} {\bf 2016}, {\em 64},~275 -- 289.

\bibitem[Chen \em{et~al.}(2017)Chen, Qi, and Wu]{Chen:2017ena}
Chen, P.; Qi, C.; Wu, L.
\newblock {Antenna placement optimisation for compressed sensing-based
  distributed MIMO radar}.
\newblock {\em IET Radar, Sonar {\&} Navigation} {\bf 2017}, {\em
  11},~285--293.

\bibitem[Yang and Xie(2016)]{yang2017SP}
Yang, Z.; Xie, L.
\newblock {Enhancing sparsity and resolution via reweighted atomic norm
  minimization}.
\newblock {\em {IEEE} Trans. Signal Process.} {\bf 2016}, {\em 64},~995--1006.

\bibitem[Shen \em{et~al.}(2016)Shen, Liu, Cui, and Wu]{Qing2016}
Shen, Q.; Liu, W.; Cui, W.; Wu, S.
\newblock {Underdetermined DOA estimation under the compressive sensing
  framework: A review}.
\newblock {\em { IEEE Access}} {\bf 2016}, {\em 4},~8865 -- 8878.

\bibitem[Yang and Xie(2016)]{Yang20166}
Yang, Z.; Xie, L.
\newblock {Exact joint sparse frequency recovery via optimization methods}.
\newblock {\em {IEEE} Trans. Signal Process.} {\bf 2016}, {\em 64},~5145 --
  5157.

\bibitem[Tipping(2001)]{Tipping2001}
Tipping, M.E.
\newblock {Sparse Bayesian Learning and the Relevance Vector Machine}.
\newblock {\em Journal of Machine Learning Research} {\bf 2001}, {\em
  1},~211--244.

\bibitem[Ji \em{et~al.}(2008)Ji, Xue, and Carin]{shihao2008}
Ji, S.; Xue, Y.; Carin, L.
\newblock {Bayesian compressive sensing}.
\newblock {\em {IEEE} Trans. Signal Process.} {\bf 2008}, {\em 56},~2346--2356.

\bibitem[Wu \em{et~al.}(2016)Wu, Zhu, and Yan]{Xiaohuan2016}
Wu, X.; Zhu, W.; Yan, J.
\newblock {Direction of arrival estimation for off-grid signals based on sparse
  Bayesian learning}.
\newblock {\em {IEEE Sensors Journal}} {\bf 2016}, {\em 16},~2004--2016.

\bibitem[Chen \em{et~al.}(2017)Chen, Zheng, Wang, Li, and Wu]{pengTSP}
Chen, P.; Zheng, L.; Wang, X.; Li, H.; Wu, L.
\newblock {Moving target detection using colocated MIMO radar on multiple
  distributed moving platforms}.
\newblock {\em {IEEE} Trans. Signal Process.} {\bf 2017}, {\em 65},~4670 --
  4683.

\bibitem[Yang \em{et~al.}(2013)Yang, Lihua, and Cishen]{yang2013}
Yang, Z.; Lihua, X.; Cishen, Z.
\newblock {Off-grid direction of arrival estimation using sparse Bayesian
  inference}.
\newblock {\em {IEEE} Trans. Signal Process.} {\bf 2013}, {\em 61},~38--43.

\bibitem[Dai \em{et~al.}(2017)Dai, Bao, Xu, and Chang]{jisheng2017}
Dai, J.; Bao, X.; Xu, W.; Chang, C.
\newblock {Root sparse Bayesian learning for off-grid DOA estimation}.
\newblock {\em {IEEE} Signal Process. Lett.} {\bf 2017}, {\em 24},~46--50.

\bibitem[Wang \em{et~al.}(2018)Wang, Zhao, Chen, and Nie]{qianli2018}
Wang, Q.; Zhao, Z.; Chen, Z.; Nie, Z.
\newblock {Grid evolution method for DOA estimation}.
\newblock {\em {IEEE} Trans. Signal Process.} {\bf 2018}, {\em 66},~2474--2383.

\bibitem[Zamani \em{et~al.}(2016)Zamani, Zayyani, and Marvasti]{Hojatollah2016}
Zamani, H.; Zayyani, H.; Marvasti, F.
\newblock {An iterative dictionary learning-based algorithm for DOA
  estimation}.
\newblock {\em {IEEE} Commun. Lett.} {\bf 2016}, {\em 20},~1784--1787.

\bibitem[Clerckx \em{et~al.}(2007)Clerckx, Craeye, Vanhoenacker-Janvier, and
  Oestges]{Clerckx2007}
Clerckx, B.; Craeye, C.; Vanhoenacker-Janvier, D.; Oestges, C.
\newblock {Impact of Antenna Coupling on $2 \times 2$ MIMO Communications}.
\newblock {\em {IEEE} Trans. Veh. Technol.} {\bf 2007}, {\em 56},~1009 --1018.

\bibitem[Zheng \em{et~al.}(2012)Zheng, Zhang, and Zhang]{Zhidong2012}
Zheng, Z.; Zhang, J.; Zhang, J.
\newblock {Joint DOD and DOA estimation of bistatic MIMO radar in the presence
  of unknown mutual coupling}.
\newblock {\em {Signal Processing}} {\bf 2012}, {\em 92},~3039 -- 3048.

\bibitem[Rocca \em{et~al.}(2017)Rocca, Hannan, Salucci, and Massa]{Paolo2017}
Rocca, P.; Hannan, M.A.; Salucci, M.; Massa, A.
\newblock {Single-snapshot DoA estimation in array antennas with mutual
  coupling through a multiscaling BCS strategy}.
\newblock {\em {IEEE} Trans. Antennas Propag.} {\bf 2017}, {\em
  65},~3203--3213.

\bibitem[Liu \em{et~al.}(2017)Liu, Zhang, Lu, Ren, and Cao]{Jianyan2017}
Liu, J.; Zhang, Y.; Lu, Y.; Ren, S.; Cao, S.
\newblock {Augmented nested arrays with enhanced DOF and reduced mutual
  coupling}.
\newblock {\em {IEEE} Trans. Signal Process.} {\bf 2017}, {\em 65},~5549 --
  5563.

\bibitem[Hawes \em{et~al.}(2017)Hawes, Mihaylova, Septer, and
  Godsill]{matthew2017}
Hawes, M.; Mihaylova, L.; Septer, F.; Godsill, S.
\newblock {Bayesian compressive sensing approaches for direction of arrival
  estimation with mutual coupling effects}.
\newblock {\em {IEEE} Trans. Antennas Propag.} {\bf 2017}, {\em
  65},~1357--1367.

\bibitem[Basikolo \em{et~al.}(2018)Basikolo, Ichige, and Arai]{Thomas2018}
Basikolo, T.; Ichige, K.; Arai, H.
\newblock {A novel mutual coupling compensation method for underdetermined
  direction of arrival estimation in nested sparse circular arrays}.
\newblock {\em {IEEE} Trans. Antennas Propag.} {\bf 2018}, {\em 66},~909 --
  917.

\bibitem[Zhang \em{et~al.}(2017)Zhang, Huang, and Liao]{ce2017}
Zhang, C.; Huang, H.; Liao, B.
\newblock {Direction finding in MIMO radar with unknown mutual coupling}.
\newblock {\em {IEEE Access}} {\bf 2017}, {\em 5},~4439 -- 4447.

\bibitem[Liao \em{et~al.}(2012)Liao, Zhang, and Chan]{liao2012}
Liao, B.; Zhang, Z.G.; Chan, S.C.
\newblock {DOA estimation and tracking of ULAs with mutual coupling}.
\newblock {\em {IEEE} Trans. Aerosp. Electron. Syst.} {\bf 2012}, {\em 48},~891
  -- 905.

\bibitem[Termos and Hochwald(2004)]{Termos:2004ch}
Termos, A.; Hochwald, B.M.
\newblock {Capacity benefits of antenna coupling}.
\newblock  2016 Information Theory and Applications (ITA); ,  2004; pp. 1--5.

\bibitem[Liu and Liao(2012)]{LIU2012517}
Liu, X.; Liao, G.
\newblock Direction finding and mutual coupling estimation for bistatic MIMO
  radar.
\newblock {\em Signal Processing} {\bf 2012}, {\em 92},~517 -- 522.

\end{thebibliography}

\end{document}